\documentclass[journal,twocolumn,twoside]{IEEEtran}

\IEEEoverridecommandlockouts
\usepackage{amsfonts}
\usepackage{color}
\usepackage{graphicx}
\usepackage[dvips]{epsfig}
\usepackage{graphics} 
\usepackage{times} 
\usepackage[cmex10]{amsmath} 
\usepackage{amssymb}  
\usepackage{cite}
\usepackage{multirow}
\usepackage[tight,footnotesize]{subfigure}
\usepackage{algorithm}

\def\inftyn #1{\left\|#1\right\|_{\infty}}
\def\twon #1{\left\|#1\right\|_2}
\def\onen #1{\left\|#1\right\|_1}

\def\frobn #1{\left\|#1\right\|_{\text{F}}}

\def\atomn #1{\left\|#1\right\|_{\cA}}
\def\atomni #1{\left\|#1\right\|_{\cA\sbra{\m{\Omega}}}}
\def\Natomni #1{\left\|#1\right\|_{\cA_N\sbra{\m{\Omega}}}}
\def\datomn #1{\left\|#1\right\|_{\cA}^*}
\def\datomni #1{\left\|#1\right\|_{\cA\sbra{\m{\Omega}}}^*}
\def\Ndatomni #1{\left\|#1\right\|_{\cA_N\sbra{\m{\Omega}}}^*}

\def\abs #1{\left|#1\right|}
\def\inp #1{\left\langle#1\right\rangle}

\def\st{\text{subject to }}

\def\bC{\mathbb{C}}

\def\bR{\mathbb{R}}

\def\bS{\mathbb{S}}

\def\m #1{\boldsymbol{#1}}

\def\cA{\mathcal{A}}

\def\cL{\mathcal{L}}

\def \qed {\hfill \vrule height6pt width 6pt depth 0pt}
\def\bee{\begin{equation}}
\def\ene{\end{equation}}

\def\beq{\begin{eqnarray}}
\def\enq{\end{eqnarray}}
\def\lentwo{\setlength\arraycolsep{2pt}}

\newtheorem{lem}{Lemma}
\newtheorem{rem}{Remark}

\newtheorem{thm}{Theorem}

\newcommand{\BOX}{\hfill\rule{2mm}{2mm}}
\newtheorem{defi}{Definition}

\def\equ #1{\begin{equation}#1\end{equation}}
\def\equa #1{\begin{eqnarray}#1\end{eqnarray}}
\def\sbra #1{\left(#1\right)}
\def\mbra #1{\left[#1\right]}
\def\lbra #1{\left\{#1\right\}}
\def\diag #1{\text{diag}#1}
\def\tr #1{\text{tr}#1}

\def\rank #1{\text{rank}#1}
\def\st {\text{ subject to }}

\title{On Gridless Sparse Methods for Line Spectral Estimation From Complete and Incomplete Data}
\author{Zai Yang, \emph{Member, IEEE}, and Lihua Xie, \emph{Fellow, IEEE}
\thanks{Manuscript November 2013; accepted by IEEE Transactions on Signal Processing March 2015.

The authors are with the School of Electrical and Electronic Engineering, Nanyang Technological University, 639798, Singapore (e-mail: \{yangzai, elhxie\}@ntu.edu.sg).}}

\begin{document}
\maketitle

\begin{abstract}
This paper is concerned about sparse, continuous frequency estimation in line spectral estimation, and focused on developing gridless sparse methods which overcome grid mismatches and correspond to limiting scenarios of existing grid-based approaches, e.g., $\ell_1$ optimization and SPICE, with an infinitely dense grid. We generalize AST (atomic-norm soft thresholding) to the case of nonconsecutively sampled data (incomplete data) inspired by recent atomic norm based techniques. We present a gridless version of SPICE (gridless SPICE, or GLS), which is applicable to both complete and incomplete data without the knowledge of noise level. We further prove the equivalence between GLS and atomic norm-based techniques under different assumptions of noise. Moreover, we extend GLS to a systematic framework consisting of model order selection and robust frequency estimation, and present feasible algorithms for AST and GLS. Numerical simulations are provided to validate our theoretical analysis and demonstrate performance of our methods compared to existing ones.
\end{abstract}

\begin{IEEEkeywords}
Line spectral estimation, atomic norm, gridless SPICE (GLS), model order selection, frequency splitting.
\end{IEEEkeywords}

\section{Introduction}
Spectral analysis of signals \cite{stoica2005spectral} is a major problem in statistical signal processing. In this paper we are concerned about the line spectral estimation problem which has wide applications in communications, radar, sonar, seismology, astronomy and so on. In particular, suppose that we observe a noisy sinusoidal signal (indexed by $j$)
\equ{y_j=\sum_{k=1}^K s_ke^{i2\pi\sbra{j-1}f_k} + e_j \label{formu:model1}}
on the index set $\mbra{M}\triangleq\lbra{1,\cdots,M}$ or a subset $\m{\Omega}\subset\mbra{M}$, where $y_j$ denotes the $j$th entry of $\m{y}\in\bC^M$ (similarly for $f_k$, $s_k$ and $e_j$), $i=\sqrt{-1}$, $f_k\in\left[0,1\right)$ and $s_k\in\bC$ denote the normalized frequency and (complex) amplitude of the $k$th sinusoidal component respectively, and $e_j\in\bC$ is the measurement noise. The sinusoid number $K<M$, usually referred to as the model order, is typically unknown in practice. Following from \cite{wang2006spectral}, the case when the signal is observed on $\mbra{M}$ is referred to as the complete data case while the other case when only samples on $\m{\Omega}\subset\mbra{M}$ are available is called the incomplete data case (or missing data case), in which the samples on the complementary set of $\m{\Omega}$, $\overline{\m{\Omega}}\triangleq\mbra{M}\backslash \m{\Omega}$, are called missing data. The missing data case is important since missing samples are common in practice that can be caused by sensor failure, outliers, weather condition or other physical constraints \cite{schafer2002missing,wang2006spectral}. Frequency estimation and model order selection are two important topics in line spectral estimation. Given $f_k$'s and $K$, $s_k$'s can be obtained by a simple least-squares method according to (\ref{formu:model1}). This paper is mainly focused on frequency estimation but we also incorporate existing model order selection tools in our methods.

Many methods have been proposed for frequency estimation. Common classical methods include periodogram (or beamforming), nonlinear least squares (NLS) and MUSIC but often have limitations (see the review in \cite{stoica2005spectral}). For example, the periodogram suffers from leakage problems and have difficulties in resolving closely separated frequencies \cite{stoica2005spectral}. It is worth noting that the recent iterative adaptive approach (IAA) \cite{yardibi2010source ,stoica2009missing} reduces the leakage of periodogram. The NLS involves nonconvex optimization and, as well as MUSIC, requires to know the model order $K$. Both the problems are not easy to deal with. Model order selection is usually a prerequisite for (or interleaved with) frequency estimation in, for example, NLS and MUSIC. Existing approaches are usually based on information-theoretic criteria or data covariance matrix such as the second order statistic of eigenvalues (SORTE) and predicted eigen-threshold approach \cite{stoica2004model,grunwald2007minimum,he2010detecting,chen1991detection}. It is recently shown in \cite{han2013improved} that SORTE outperforms other methods in a related problem.

With the development of sparse signal representation (SSR) and later the compressed sensing (CS) concept \cite{donoho2006compressed}, sparse methods for frequency estimation have been popular in the past decade. In this kind of methods, the continuous frequency domain $\left[0,1\right)$ is discretized/gridded into a finite set of grid points. By assuming that the true frequencies are on (practically, close to) some grid points, the observation model is approximately written into a linear system of equations. Then frequency estimation is accomplished by sparse signal recovery followed by support detection. Two prominent sparse methods are $\ell_1$ optimization and sparse iterative covariance-based estimation (SPICE) \cite{malioutov2005sparse,stoica2011new,stoica2011spice,stoica2012spice}. SPICE is usually more practical since it estimates the noise variance, which is unavailable in advance, jointly with frequency estimation.

Since CS so far has been focused on signals that can be sparsely represented under a finite dictionary (or a finite set of atoms), discretization/gridding of the frequency domain is inevitable in early sparse methods. According to the wisdom of CS the sampling grid should not be too dense, otherwise almost complete correlations between adjacent atoms (or steering vectors) may degrade the sparse recovery performance. However, it is intuitively reasonable and in fact has been verified by many algorithms that a dense grid leads to a more accurate frequency estimate since both grid mismatches (between grid points and the true frequencies) and approximation errors (of the observation model) can be reduced with a dense grid. Therefore, one naturally wonders whether the existing sparse methods can be practically implemented with an infinitely dense grid or equivalently, directly on the continuous interval $\left[0,1\right)$ without gridding and, if implementable, what performances the \emph{gridless} sparse methods can obtain. This paper will answer these questions.

Before proceeding to gridless sparse methods, it is worth noting that grid-based methods have been proposed to alleviate the drawbacks of the finite discretization with affordable computational workloads. Many of them start with a coarse grid and gradually modify the frequency estimate out of or during the algorithms. Examples include iterative grid refinement \cite{malioutov2005sparse} and joint sparse signal and parameter estimation \cite{shutin2011sparse, hu2012compressed,yang2012robustly, yang2013off,austin2013dynamic,hu2013fast, shutin2013incremental,tan2014joint}, where \cite{shutin2011sparse,hu2012compressed,shutin2013incremental} are sparse versions of the space-alternating generalized expectation-maximization (SAGE) algorithm \cite{fleury1999channel,feder1988parameter}. Since the observed samples are nonlinear functions of the frequencies by (\ref{formu:model1}), the joint estimation methods typically need to carry out nonconvex optimization and cannot guarantee global optimality. Other methods such as \cite{duarte2013spectral,fannjiang2012coherence} start with a fixed, highly dense grid and iteratively optimize sparse solutions supported on sufficiently separate grid points.

The first gridless sparse method for frequency estimation is introduced in \cite{candes2013towards} motivated by the concept of atomic norm (or total variation norm) for continuous-time signals \cite{aleksanyan1944real,chandrasekaran2012convex}, which generalizes the $\ell_1$ norm for the discrete counterpart. Therefore, the atomic norm-based methods in \cite{candes2013towards} and later papers \cite{candes2013super, bhaskar2013atomic,tang2012compressed} correspond to gridless versions (or limiting scenarios with an infinitely dense grid) of the $\ell_1$-based methods. In particular, the noiseless complete data case is studied in \cite{candes2013towards}, where it is shown that the frequencies can be exactly recovered provided that they are appropriately separated. The bounded-energy-noise case is then studied in \cite{candes2013super}. An atomic norm soft thresholding (AST) method is presented in \cite{bhaskar2013atomic} in the presence of stochastic noise, a common assumption in the literature. In the presence of missing data, the noiseless case is studied in \cite{tang2012compressed} via atomic norm minimization with exact recovery proven under some technical assumptions. Since computation of the atomic norm can be formulated as convex programming \cite{candes2013towards,bhaskar2013atomic}, the gridless sparse methods above can be solved in a polynomial time. Other related papers include \cite{azais2014spike} for complete data and \cite{chen2014robust} based on matrix completion for incomplete data. After submission of this paper, atomic norm methods have also been proposed in the case of multiple measurement vectors encountered in array processing and for further enhancing resolution \cite{yang2014continuous,yang2014exact,yang2015achieving}.


In this paper, we develop new gridless sparse methods for line spectral estimation and demonstrate their relations to the existing grid-based methods. Note that 1) atomic norm-based methods are still absent for noisy incomplete data, and 2) the existing atomic norm-based methods require the practically unknown noise variance/energe. An estimate can be possibly obtained as in \cite{bhaskar2013atomic} in the complete data case, however, it is not clear how to do this with incomplete data. The contributions of this paper are summarized as follows:
\begin{enumerate}
 \item We generalize AST and its theoretical results in \cite{bhaskar2013atomic} to the missing data case.
 \item We develop the gridless version of SPICE, named as gridless SPICE or GLS for short. GLS is obtained based on our recent work \cite{yang2014discretization} where the focus is on the \emph{spatial} spectral analysis (a.k.a. array processing) as opposed to the \emph{temporal} spectral analysis considered here. Moreover, we extend it to a systematic framework for line spectral estimation consisting of model order selection and improved frequency estimation.
 \item We explore connections between GLS and atomic norm-based methods and prove their equivalence under different assumptions of noise. The result holds in both the complete and missing data cases.
 \item We develop feasible algorithms for AST and GLS based on duality and the alternating direction method of multipliers (ADMM) \cite{boyd2011distributed}.
 \item We demonstrate that existing grid-based SPICE and $\ell_1$ optimization are approximate versions of GLS analytically and via numerical simulations.
\end{enumerate}

Notations used in this paper are as follows. $\bR$ and $\bC$ denote the sets of real and complex numbers respectively. Boldface letters are reserved for vectors and matrices. For an integer $N$, $[N]\triangleq\lbra{1,\cdots,N}$. $\abs{\cdot}$ denotes the amplitude of a scalar or cardinality of a set. $\onen{\cdot}$, $\twon{\cdot}$ and $\frobn{\cdot}$ denote the
$\ell_1$, $\ell_2$ and Frobenius norms respectively. $\m{A}^T$ and $\m{A}^H$ are the matrix transpose and conjugate transpose of $\m{A}$ respectively. $x_j$ is the $j$th entry of a vector $\m{x}$. Unless otherwise stated, $\m{x}_{\m{\Omega}}$ and $\m{A}_{\m{\Omega}}$ respectively reserve the entries of $\m{x}$ and the rows of $\m{A}$ in the index set $\m{\Omega}$. For a vector $\m{x}$, $\diag\sbra{\m{x}}$ is a diagonal matrix with $\m{x}$ being its diagonal. $\m{x}\succeq\m{0}$ means $x_j\geq0$ for all $j$. $\tr\sbra{\m{A}}$ denotes the trace of a matrix $\m{A}$. For positive semidefinite matrices $\m{A}$ and $\m{B}$, $\m{A}\geq\m{B}$ means that $\m{A}-\m{B}$ is positive semidefinite. $E\mbra{\cdot}$ denotes expectation and $\widehat{f}$ is an estimator of $f$. For notational simplicity, a random variable and its numerical value will not be distinguished.

The rest of the paper is organized as follows. Section \ref{sec:preliminary} introduces some preliminary results. Section \ref{sec:AST_inc} extends AST to the missing data case. Section \ref{sec:GLS} presents GLS and Section \ref{sec:framework} extends it to a systematic framework for line spectral estimation. Section \ref{sec:connection} proves the equivalence between GLS and atomic norm-based methods. Section \ref{sec:ADMM} presents feasible algorithms for AST and GLS. Section \ref{sec:simulation} provides numerical simulations and Section \ref{sec:conclusion} concludes this paper.

\section{Preliminaries} \label{sec:preliminary}

\subsection{$\ell_1$ Norm Denoising}
Consider the problem of recovering a signal $\m{z}\in\bC^M$ from its noisy measurement $\m{y}\in\bC^M$, with the prior knowledge that $\m{z}$ has a sparse representation under a discrete dictionary $\m{A}\in\bC^{M\times N}$, i.e., there exists a sparse vector $\m{s}\in\bC^{N}$ such that $\m{z}=\m{A}\m{s}$. The $\ell_1$ norm has been widely used for this signal denoising problem. In particular, $\m{z}$ is recovered by solving $\m{s}$ from the following optimization problem:
\equ{\min_{\m{s}} \mu \onen{\m{s}}+g\sbra{\m{y}-\m{A}\m{s}}, \label{formu:l1ND}}
where function $g\sbra{\cdot}$ plays data fitting and the regularization parameter $\mu>0$ balances the fidelity of the measurement $\m{y}$ and the sparsity of $\m{s}$. Collectively, we call (\ref{formu:l1ND}) $\ell_1$ norm denoising (L1ND). Three common choices of $g\sbra{\cdot}$ are $\twon{\cdot}^2$, $\twon{\cdot}$ and $\onen{\cdot}$, with which (\ref{formu:l1ND}) is referred to as Lasso, square root- (SR-) Lasso and least absolute deviation- (LAD-) Lasso, respectively \cite{tibshirani1996regression,belloni2011square,wang2007robust}. From a statistical perspective, Lasso suits for Gaussian noise and LAD-Lasso is robust to outliers. Compared to Lasso, SR-Lasso requires loose assumptions of the noise distribution with an easy choice of $\mu$ \cite{belloni2011square}.

\subsection{Atomic Norm} \label{sec:atomicnorm}
The concept of atomic norm is introduced in \cite{chandrasekaran2012convex}, which generalizes many common sparse norms such as the $\ell_1$ norm and the nuclear norm of matrices. Let $\cA$ be a collection of atoms satisfying that its convex hull, $\text{conv}\sbra{\cA}$, is compact, centrally symmetric, and contains the origin as an interior point. Then the gauge function of $\text{conv}\sbra{\cA}$ defines a norm which is called the atomic norm and denoted by $\atomn{\cdot}$:
\equ{\begin{split}\atomn{\m{y}}
&\triangleq\inf\lbra{t>0: \m{y}\in t\text{conv}\sbra{\cA}} \\
&= \inf\lbra{\sum_k c_k: \m{y}=\sum_k c_k\m{a}_k, c_k\geq0, \m{a}_k\in\cA}.\end{split} \label{formu:atomicnorm}}
The dual norm of the atomic norm is given by
\equ{\datomn{\m{z}}=\sup\lbra{\inp{\m{z},\m{a}}_{\bR}: \atomn{\m{a}}\leq1},}
where $\inp{\m{z},\m{a}}_{\bR}=\Re\inp{\m{z},\m{a}}=\Re\lbra{\m{a}^H\m{z}}$ and $\Re$ takes the real part of a complex number. Moreover, it can be shown that $\text{conv}\sbra{\cA}=\lbra{\m{a}:\atomn{\m{a}}\leq1}$ and thus $\cA$ contains all extreme points of $\lbra{\m{a}:\atomn{\m{a}}\leq1}$. It follows that
\equ{\datomn{\m{z}}=\sup_{\m{a}\in\cA}\inp{\m{z},\m{a}}_{\bR}.\label{formu:dualatomnorm}}

\subsection{AST for Line Spectral Estimation From Complete Data} \label{sec:AST}
The observation model in (\ref{formu:model1}) can be written more compactly as follows:
\equ{\m{y}=\sum_{k=1}^K \m{a}\sbra{f_k}s_k + \m{e} =\m{A}\sbra{\m{f}}\m{s}+\m{e}, \label{formu:model2}}
where $\m{a}\sbra{f_k}=\mbra{1,e^{i2\pi f_k},\cdots,e^{i2\pi\sbra{M-1}f_k}}^T\in\bC^{M}$, $\m{A}\sbra{\m{f}}=\mbra{\m{a}\sbra{f_1},\dots,\m{a}\sbra{f_K}}\in\bC^{M\times K}$, $\m{y}\in\bC^M$ is a vector by stacking all $y_j$, and $\m{s}\in\bC^K$, $\m{e}\in\bC^M$ are similarly defined. Denote $\m{a}\sbra{f,\phi}=\m{a}\sbra{f}\phi$, where $\phi\in\bS^1\triangleq\lbra{\phi\in\bC:\abs{\phi}=1}$. The set of atoms $\cA$ in this application is defined as
\equ{\cA\triangleq\lbra{\m{a}\sbra{f,\phi}: f\in\left[0,1\right), \phi\in\bS^1}. \label{formu:atomset1}}
The induced atomic norm $\atomn{\cdot}$ can be computed via semidefinite programming (SDP) \cite{bhaskar2013atomic}:
\equ{\atomn{\m{y}}=\min_{x,\m{u}} \frac{1}{2}\sbra{x+ u_1}, \st \begin{bmatrix}x& \m{y}^H \\ \m{y} & T\sbra{\m{u}} \end{bmatrix}\geq\m{0}, \label{formu:SDP_ANM}}
where $\m{u}\in\bC^M$ and $T\sbra{\m{u}}\in\bC^{M\times M}$ denotes a (Hermitian) Toeplitz matrix with
\equ{T\sbra{\m{u}}=\begin{bmatrix}u_1 & u_2 & \cdots & u_M\\ {u}_2^H & u_1 & \cdots & u_{M-1}\\ \vdots & \vdots & \ddots & \vdots \\ {u}_M^H & {u}_{M-1}^H & \cdots & u_1\end{bmatrix}, \label{formu:Toeplitz}}
where $u_j$ denotes the $j$th entry of $\m{u}$.

In the presence of independently and identically distributed (i.i.d.) zero-mean Gaussian noise with noise variance $\sigma_0$, \cite{bhaskar2013atomic} proposes the following atomic soft thresholding (AST) method for estimating the noiseless sinusoidal signal $\m{z}\triangleq\m{A}\sbra{\m{f}}\m{s}$:
\equ{\min_{\m{z}} \mu\atomn{\m{z}} +\frac{1}{2}\twon{\m{y}-\m{z}}^2, \label{formu:AST}}
where $\mu\approx\sqrt{M\ln M}\sigma_0^{\frac{1}{2}}$ when $M$ is sufficiently large. (\ref{formu:AST}) can be formulated as the following SDP by (\ref{formu:SDP_ANM}):
\equ{\min_{x,\m{u},\m{z}} \frac{\mu}{2}\sbra{x+ u_1} +\frac{1}{2}\twon{\m{y}-\m{z}}^2, \st \begin{bmatrix}x& \m{z}^H \\ \m{z} & T\sbra{\m{u}} \end{bmatrix}\geq\m{0}. \label{formu:AST11}}
Given the optimal solution $\sbra{x^*,\m{u}^*,\m{z}^*}$ of (\ref{formu:AST11}), the frequency and amplitude estimates $\widehat{\m{f}}$ and $\widehat{\m{s}}$ can be obtained from the Vandermonde decomposition of $T\sbra{\m{u}^*}$ (see Lemma \ref{lem:toeplitz} in Appendix \ref{sec:retrieval}). In particular, it holds that $T\sbra{\m{u}^*}=\m{A}\sbra{\widehat{\m{f}}}\diag\sbra{\abs{\widehat{\m{s}}}} \m{A}^H\sbra{\widehat{\m{f}}}$ and $\m{z}^*=\m{A}\sbra{\widehat{\m{f}}}\widehat{\m{s}}$, where $\abs{\cdot}$ operates elementwise for a vector (the two $\widehat{\m{s}}$ in the two equations above are identical following from the proof of Proposition II.1 in \cite{tang2012compressed}). We provide a computational method of the Vandermonde decomposition in Appendix \ref{sec:retrieval} which will be revisited later.

\section{AST for Incomplete Data} \label{sec:AST_inc}

\subsection{Atomic Norm for Incomplete Data}
Suppose that the observed samples are on a subset $\m{\Omega}\subset\mbra{M}$, where $\m{\Omega}$ is assumed to be sorted ascendingly. Denote the sample size $L=\abs{\m{\Omega}}\leq M$ and the range of the sampling period $\overline{M}=\Omega_L-\Omega_1+1\leq M$. Note that $\overline{M}$ is more practically relevant than $M$ since we can always re-index the observed samples by the set $\m{\Omega}-\Omega_1+1\triangleq\lbra{\Omega_l-\Omega_1+1: l\in\mbra{L}} =\lbra{1,\Omega_2-\Omega_1+1,\dots,\overline{M}}$. We define the set of atoms in this missing data case as follows:
\equ{\begin{split}\cA\sbra{\m{\Omega}}
&\triangleq\lbra{\m{a}_{\m{\Omega}}: \m{a}\in\cA}\\
&= \lbra{\m{a}_{\m{\Omega}}\sbra{f,\phi}: f\in\left[0,1\right), \phi\in\bS^1},\end{split}}
where $\m{a}_{\m{\Omega}}\sbra{f,\phi}=\m{a}_{\m{\Omega}}\sbra{f}\phi$ and $\m{a}_{\m{\Omega}}\sbra{f}$ is a subvector of $\m{a}\sbra{f}$ indexed by $\m{\Omega}$.
The convex hull $\text{conv}\sbra{\cA\sbra{\m{\Omega}}}$ can be shown to satisfy the conditions specified in Subsection \ref{sec:atomicnorm}. It follows that the gauge function of $\text{conv}\sbra{\cA\sbra{\m{\Omega}}}$ defines a norm that is denoted by $\atomni{\cdot}$.


\begin{lem} For the atomic norm $\atomni{\cdot}$ it holds that
\equ{\begin{split}
&\atomni{\m{y}_{\m{\Omega}}}=\min_{\m{y}_{\overline{\m{\Omega}}}}\atomn{\m{y}}\\
&=\min_{x,\m{u}, \m{y}_{\overline{\m{\Omega}}}} \frac{1}{2}\sbra{x+ u_1}, \st \begin{bmatrix}x& \m{y}^H \\ \m{y} & T\sbra{\m{u}} \end{bmatrix}\geq\m{0}.\end{split} \label{formu:SDP_AN_inc}} \label{thm:SDP_AN_inc}
\end{lem}
\begin{proof} By definition the following equalities hold:
\equ{\begin{split}
&\min_{\m{y}_{\overline{\m{\Omega}}}}\atomn{\m{y}}\\
&= \min_{\m{y}_{\overline{\m{\Omega}}}}\inf\lbra{\sum_k c_k: \m{y}=\sum_k c_k\m{a}_k, c_k\geq0, \m{a}_k\in\cA}\\
&= \inf\lbra{\sum_k c_k: \m{y}_{\m{\Omega}}=\sum_k c_k\sbra{\m{a}_k}_{\m{\Omega}}, c_k\geq0, \m{a}_k\in\cA}\\
&= \inf\lbra{\sum_k c_k: \m{y}_{\m{\Omega}}=\sum_k c_k\m{b}_k, c_k\geq0, \m{b}_k\in\cA\sbra{\m{\Omega}}}\\
&= \atomni{\m{y}_{\m{\Omega}}}.
\end{split}}
The second equality in (\ref{formu:SDP_AN_inc}) follows from (\ref{formu:SDP_ANM}).
\end{proof}

Note that the SDP formulation in (\ref{formu:SDP_AN_inc}) has been studied in \cite{tang2012compressed} for exact frequency recovery in the noiseless missing data case. Lemma \ref{thm:SDP_AN_inc} shows that this technique is exactly computing the atomic norm of the incomplete data.

For the dual atomic norm we have similarly to (\ref{formu:dualatomnorm}) that
\equ{\datomni{\m{z}_{\m{\Omega}}}=\sup_{f,\phi\in\bS^1}\inp{\m{z}_{\m{\Omega}}, \m{a}_{\m{\Omega}}\sbra{f,\phi}}_{\bR} = \sup_{f}\abs{\inp{\m{z}_{\m{\Omega}}, \m{a}_{\m{\Omega}}\sbra{f}}}, \label{formu:datomni}}
which will be useful in later analysis.

\subsection{AST for Incomplete Data}
Suppose that the observed samples are contaminated with i.i.d. noise. As suggested by \cite{bhaskar2013atomic} we estimate the noiseless signal, denoted by $\m{z}$ (or $\m{z}_{\m{\Omega}}$ on $\m{\Omega}$), by solving the following AST problem:
\equ{\min_{\m{z}_{\m{\Omega}}} \mu\atomni{\m{z}_{\m{\Omega}}} +\frac{1}{2}\twon{\m{y}_{\m{\Omega}}-\m{z}_{\m{\Omega}}}^2, \label{formu:AST_inc}}
where $\mu>0$ is to be specified.
Following from (\ref{formu:SDP_AN_inc}), (\ref{formu:AST_inc}) can be written into the following SDP:
\equ{\begin{split}
&\min_{x,\m{u},\m{z}} \frac{\mu}{2}\sbra{x+u_1} +\frac{1}{2}\twon{\m{y}_{\m{\Omega}}-\m{z}_{\m{\Omega}}}^2, \\
&\st \begin{bmatrix}x& \m{z}^H \\ \m{z} & T\sbra{\m{u}} \end{bmatrix}\geq\m{0}. \end{split} \label{formu:AST_SDP_inc}}

\begin{thm} Suppose the signal $\m{y}$ given by (\ref{formu:model1}) or (\ref{formu:model2}) is observed on the subset $\m{\Omega}\subset\mbra{M}$. Denote the original noiseless signal by $\m{z}^o=\m{A}\sbra{\m{f}}\m{s}$. The estimate $\widehat{\m{z}}$ of $\m{z}^o$ given by the solution of AST in (\ref{formu:AST_inc}) or (\ref{formu:AST_SDP_inc}) with $\mu\geq E\datomni{\m{e}_{\m{\Omega}}}$ has the expected (per-element) mean squared error (MSE)
\equ{\frac{1}{L} E\twon{\widehat{\m{z}}_{\m{\Omega}} - \m{z}_{\m{\Omega}}^o}^2\leq \frac{\mu}{L} \sum_{k=1}^K\abs{s_k}.}
Moreover, assume that $\m{e}_{\m{\Omega}}$ denotes i.i.d. zero-mean Gaussian noise with noise variance $\sigma_0$. Then the expected dual norm is upper bounded as follows:
\equ{E\datomni{\m{e}_{\m{\Omega}}}\leq \mu^*\triangleq \min_{p>1}\frac{p}{p-1}\sqrt{L\sbra{\ln \overline{M}+\ln\sbra{\pi p}+1}}\sigma_0^{\frac{1}{2}}, \label{formu:upperbound}}
where the optimizer $p^*$ satisfies that $2\ln\overline{M}<p^*<5\ln\overline{M}$ as $\overline{M}\geq100$. \label{thm:AST_inc}
\end{thm}
\begin{proof} The first part of the theorem is a direct result of \cite[Theorem 1]{bhaskar2013atomic}. The upper bound of the expected dual norm in the case of i.i.d. Gaussian noise is derived in Appendix \ref{sec:proofupperbound}.
\end{proof}

It is interesting to note that Theorem \ref{thm:AST_inc} generalizes the result in the complete data case stated in \cite[Theorem 2]{bhaskar2013atomic}, where $\m{\Omega}=\mbra{M}$ and $L=\overline{M}=M$. The upper bound $\mu^*\approx\sqrt{L\ln \overline{M}}\sigma_0^{\frac{1}{2}}$ holds when $\overline{M}$ is sufficiently large, which depends on the range $\overline{M}$ of $\m{\Omega}$ besides the sample size $L$. By Theorem \ref{thm:AST_inc} AST with $\mu=\mu^*$ guarantees to produce a consistent signal estimate (on $\m{\Omega}$) if $K=o\sbra{\sqrt{\frac{L}{\ln M}}}$. In the limiting noiseless case where $\mu^*\propto\sigma_0^{\frac{1}{2}}\rightarrow0$, AST in (\ref{formu:AST_inc}) is equivalent to computing $\atomni{\m{y}_{\m{\Omega}}}$, which has been shown in \cite{tang2012compressed} to result in exact frequency recovery when the frequencies are sufficiently separate. Consequently it is expected that accurate frequency estimation can be obtained based on (\ref{formu:AST_inc}), which can be extracted from the Vandermonde decomposition of $T\sbra{\m{u}^*}$ as in the complete data case. Note, however, that it is not easy to determine the regularization parameter $\mu^*$ in practical scenarios where the noise variance $\sigma$ is unavailable and difficult to estimate from the incomplete data set. Motivated by this observation, we will turn to SPICE in Section \ref{sec:GLS}.

\subsection{Extension: Gridless Atomic Norm vs. Grid-based $\ell_1$ Norm} \label{sec:gridlessvsgrid}
In this subsection we show rigorously that the gridless atomic norm is the limiting scenario of the grid-based $\ell_1$ norm as the grid gets infinitely dense. In particular, suppose that a uniform grid of $N$ points $\widetilde{\m{f}}\triangleq\lbra{0,\frac{1}{N},\cdots,1-\frac{1}{N}}$ is used to sample the continuous domain $\left[0,1\right)$, or equivalently, the frequency variables are constrained on $\widetilde{\m{f}}$. Denote the resulting discrete set of atoms $\cA_N\sbra{\m{\Omega}}=\lbra{\m{a}_{\m{\Omega}}\sbra{f,\phi},f\in\widetilde{\m{f}}, \phi\in\bS^1}$ and its induced atomic norm by $\Natomni{\cdot}$. Moreover, define $\m{A}_{\m{\Omega}}=\mbra{\m{a}_{\m{\Omega}} \sbra{\widetilde{f}_1}, \m{a}_{\m{\Omega}}\sbra{\widetilde{f}_2},\dots, \m{a}_{\m{\Omega}}\sbra{\widetilde{f}_N}}$. We obtain by the definition of the atomic norm that
\equ{\begin{split}&\Natomni{\m{y}_{\m{\Omega}}}\\
&=\min_{c_j\geq 0,\phi_j\in\bS^1} \sum_{j=1}^Nc_j, \st \m{y}_{\m{\Omega}}=\sum_{j=1}^N c_j\m{a}_{\m{\Omega}}\sbra{\widetilde{f}_j,\phi_j}\\ &=\min_{\m{s}}\onen{\m{s}} \st \m{A}_{\m{\Omega}}\m{s}=\m{y}_{\m{\Omega}}, \end{split} \label{formu:linkgrid}}
where $s_j=c_j\phi_j$. Note that the grid-based atomic norm $\Natomni{\cdot}$ is linked to the $\ell_1$ norm by (\ref{formu:linkgrid}). Correspondingly, in the missing data case the L1ND formulation in (\ref{formu:l1ND}) can be equivalently written into
\equ{\min_{\m{z}_{\m{\Omega}}}\mu\Natomni{\m{z}_{\m{\Omega}}}+g\sbra{\m{y}_{\m{\Omega}}-\m{z}_{\m{\Omega}}}, \label{formu:NAND}}
where $\m{z}_{\m{\Omega}}\triangleq\m{A}_{\m{\Omega}}\m{s}$. Intuitively, as $N\rightarrow+\infty$ the discrete atomic set $\cA_N\sbra{\m{\Omega}}$ becomes the continuous atomic set $\cA\sbra{\m{\Omega}}$. Therefore, $\Natomni{\m{z}_{\m{\Omega}}}\rightarrow\atomni{\m{z}_{\m{\Omega}}}$. Formally, we have the following result:
\equ{\sbra{1-\frac{\pi \overline{M}}{N}}\Natomni{\m{z}_{\m{\Omega}}} \leq \atomni{\m{z}_{\m{\Omega}}}\leq \Natomni{\m{z}_{\m{\Omega}}}
\label{formu:Natomnvsatomn}}
which generalizes the result in the complete data case in \cite{bhaskar2013atomic} and is proven in Appendix \ref{sec:Natomnvsatomn}. Two implications of (\ref{formu:Natomnvsatomn}) are as follows: 1) the atomic norm is the limiting scenario of the $\ell_1$ norm when the grid becomes infinitely dense, and 2) the grid-based $\ell_1$ optimization methods are good approximations of corresponding gridless atomic norm methods when the grid size $N=O\sbra{M}$.

\section{Gridless SPICE (GLS)} \label{sec:GLS}
\subsection{Introduction to SPICE}
SPICE \cite{stoica2011new,stoica2011spice,stoica2012spice} is a grid-based sparse method for line spectral estimation based on weighted covariance fitting (WCF). It is advantageous to Lasso in practice in the sense that it automatically estimates the noise variance which is typically unavailable in advance. To develop SPICE, it is assumed that the phases of $s_k$ in (\ref{formu:model2}) or (\ref{formu:model1}), $k\in\mbra{K}$, are independently and uniformly distributed, which is a common assumption in covariance-based methods, e.g., MUSIC, and also in \cite{tang2012compressed}. It follows that $E\mbra{\m{s}\m{s}^H}=\diag\sbra{\abs{s_k}^2}\triangleq\diag\sbra{\m{p}}$, where $\m{p}$ is called the power parameter. Further assume that the noise $\m{e}$ is independent from $\m{s}$ and satisfies that $E\mbra{\m{e}\m{e}^H}=\diag\sbra{\m{\sigma}}$, where $\m{\sigma}$ denotes the noise variance parameter whose elements can be different from each other. Then the covariance matrix of $\m{y}$ has the following expression:
\equ{\m{R}=E\mbra{\m{y}\m{y}^H} =\m{A}\sbra{\m{f}}\diag\sbra{\m{p}}\m{A}^H\sbra{\m{f}}+\diag\sbra{\m{\sigma}}. \label{formu:R}}
SPICE attempts to minimize a WCF criterion as follows:
\equ{\begin{split}h\sbra{\m{f},\m{p},\m{\sigma}}
&=\frobn{\m{R}^{-\frac{1}{2}}\sbra{\m{y}\m{y}^H-\m{R}}}^2\\
& =\tr\sbra{\m{R}}+\twon{\m{y}}^2\m{y}^H\m{R}^{-1}\m{y}-2\twon{\m{y}}^2. \end{split} \label{formu:criterion2}}
According to \cite{stoica2011new,stoica2011spice} and references therein, (\ref{formu:criterion2}) is a \emph{suboptimal} criterion, a new understanding of which will be provided in Subsection \ref{sec:subopt}.
Note that the resulting optimization problem
\equ{\min_{\m{f},\m{p}\succeq\m{0},\m{\sigma}\succeq\m{0}} \tr\sbra{\m{R}}+\twon{\m{y}}^2\m{y}^H\m{R}^{-1}\m{y} \label{formu:GLS1}}
is nonconvex since the data covariance $\m{R}$ is nonlinear with respect to $\m{f}$ by (\ref{formu:R}). Like other existing sparse methods, discretization is applied to the continuous frequency domain to eliminate the dependence of $\m{R}$ on $\m{f}$. An alternating algorithm, named as SPICE, is then developed to solve the grid-based version of (\ref{formu:GLS1}). It is shown in \cite{rojas2013note,babu2014connection} that SPICE is connected to $\ell_1$ optimization methods, which will be revisited later.

\subsection{GLS in the Complete Data Case} \label{sec:GLS_com}
We now introduce the gridless version of SPICE, namely, GLS, in the complete data case. In particular, GLS adopts the WCF criterion of SPICE in (\ref{formu:criterion2}) but exactly solves (\ref{formu:GLS1}). Rather than the discretization which linearizes the covariance matrix $\m{R}$, a critical technique of GLS is to reparameterize $\m{R}$ by introducing a positive semidefinite Toeplitz matrix $T\sbra{\m{u}}=\m{A}\sbra{\m{f}}\diag\sbra{\m{p}}\m{A}^H\sbra{\m{f}}$. It follows from (\ref{formu:R}) that
\equ{\m{R}=T\sbra{\m{u}}+\diag\sbra{\m{\sigma}} \label{formu:reparameterization}}
with $T\sbra{\m{u}}\geq\m{0}$ and $\m{\sigma}\succeq\m{0}$. GLS is based on the following result.
\begin{lem} The two representations of $\m{R}$ in (\ref{formu:R}) and (\ref{formu:reparameterization}) are equivalent in the sense that, if $\m{R}$ can be represented by one, then it can be represented by the other. \label{lem:equivalenceR}
\end{lem}
\begin{proof} It is a direct result of the Vandermonde decomposition lemma (see Lemma \ref{lem:toeplitz} in Appendix \ref{sec:retrieval}).
\end{proof}

By Lemma \ref{lem:equivalenceR}, the optimization problem in (\ref{formu:GLS1}) of GLS is equivalent to the following SDP:
\equ{\begin{split}
&\min_{\m{u},\m{\sigma}\succeq\m{0}} \tr\sbra{\m{R}}+\twon{\m{y}}^2\m{y}^H\m{R}^{-1}\m{y}, \st T\sbra{\m{u}}\geq\m{0}\\
&=\min_{x,\m{u},\m{\sigma}\succeq\m{0}} \tr\sbra{\m{R}}+\twon{\m{y}}^2x,\\
&\quad\st \begin{bmatrix}x & \m{y}^H \\ \m{y} & \m{R}\end{bmatrix}\geq\m{0} \text{ and } T\sbra{\m{u}}\geq\m{0},\end{split} \label{formu:GLS2}}
where $\m{R}$ is given in (\ref{formu:reparameterization}).
After (\ref{formu:GLS2}) is solved, the remaining task is to retrieve from its solution $\sbra{\m{u}^*,\m{\sigma}^*}$ the parameter estimate $\sbra{\widehat{\m{f}}, \widehat{\m{p}}, \widehat{\m{\sigma}}}$ of interest or solution  of (\ref{formu:GLS1}). In particular, if $T\sbra{\m{u}^*}$ is rank-deficient, then $\widehat{\m{f}}$ and $\widehat{\m{p}}$ can be uniquely determined by its Vandermonde decomposition using Lemma \ref{lem:toeplitz} in Appendix \ref{sec:retrieval}, and $\widehat{\m{\sigma}}=\m{\sigma}^*$. However, if $T\sbra{\m{u}^*}$ has full rank, then $\widehat{\m{f}}$ and $\widehat{\m{p}}$ cannot be uniquely determined. That means, (\ref{formu:GLS1}) has multiple optimal solutions. Among these solutions, we choose the one such that $\widehat{\m{f}}$ and $\widehat{\m{p}}$ have the minimum length since it is always of interest to simplify the model. In particular, let $\delta=\lambda_{\text{min}}\sbra{T\sbra{\m{u}^*}}$ be the minimum eigenvalue of $T\sbra{\m{u}^*}$. Then $\widehat{\m{f}}$ and $\widehat{\m{p}}$, which satisfy that $\abs{\widehat{\m{f}}}=\abs{\widehat{\m{p}}}\leq M-1$, are uniquely obtained by the Vandermonde decomposition of $T\sbra{\m{u}^*}-\delta\m{I}$, and $\widehat{\m{\sigma}} =\m{\sigma}^*+\delta\m{1}$, where $\m{I}$ and $\m{1}$ are respectively an identity matrix and a vector of ones.

\begin{rem} Under the assumption of homoscedastic noise, the data covariance matrix $\m{R}$ itself is Toeplitz and thus can be represented as $\m{R}=T\sbra{\widetilde{\m{u}}}$, where $\widetilde{\m{u}}\in\bC^M$ and $T\sbra{\widetilde{\m{u}}}\geq\m{0}$. Then the SDP of GLS in (\ref{formu:GLS2}) can be simplified accordingly (by simply setting $\m{\sigma}=\m{0}$). After the SDP is solved, the parameter estimate $\sbra{\widehat{\m{f}}, \widehat{\m{p}}, \widehat{\m{\sigma}}}$ can be given in the same manner. \label{rem:Rtoeplitz}
\end{rem}

\subsection{GLS in the Missing Data Case}
In the missing data case, only samples on $\m{\Omega}\subset\mbra{M}$ are observed. The same WCF criterion is adopted but applied only to the available data $\m{y}_{\m{\Omega}}$. Under the same assumptions as in SPICE, the covariance matrix of $\m{y}_{\m{\Omega}}$, denoted by $\m{R}_{\m{\Omega}}$, is
\equ{\m{R}_{\m{\Omega}}=E\mbra{\m{y}_{\m{\Omega}}\m{y}_{\m{\Omega}}^H} =\m{A}_{\m{\Omega}}\sbra{\m{f}}\diag\sbra{\m{p}}\m{A}_{\m{\Omega}}^H\sbra{\m{f}} +\diag\sbra{\m{\sigma}_{\m{\Omega}}}. \label{formu:R_inc}}
Note that $\m{A}_{\m{\Omega}}\sbra{\m{f}}=\m{\Gamma}_{\m{\Omega}}\m{A}\sbra{\m{f}}$, where $\m{\Gamma}_{\m{\Omega}}\in \lbra{0,1}^{L\times M}$ and its elements equal 1 only at $\sbra{j, \Omega_j}$, $j\in\mbra{L}$. Consequently, $\m{R}_{\m{\Omega}}$ can be equivalently reparameterized as
\equ{\m{R}_{\m{\Omega}}=\m{\Gamma}_{\m{\Omega}}T\sbra{\m{u}}\m{\Gamma}_{\m{\Omega}}^T+\diag\sbra{\m{\sigma}_{\m{\Omega}}} \label{formu:R_Omega}}
under the constraints $T\sbra{\m{u}}\geq\m{0}$ and $\m{\sigma}_{\m{\Omega}}\succeq\m{0}$, where $T\sbra{\m{u}}$ can be interpreted as the covariance of the ``clean'' complete data as before.
So GLS solves the following convex optimization problem:
\equ{\min_{\m{u},\m{\sigma}_{\m{\Omega}}\succeq\m{0}} \tr\sbra{\m{R}_{\m{\Omega}}}+\twon{\m{y}_{\m{\Omega}}}^2 \m{y}_{\m{\Omega}}^H\m{R}_{\m{\Omega}}^{-1}\m{y}_{\m{\Omega}},\\
\st T\sbra{\m{u}}\geq\m{0}, \label{formu:GLS2_inc}}
where $\m{R}_{\m{\Omega}}$ is given in (\ref{formu:R_Omega}).
Given the solution $\sbra{\m{u}^*,\m{\sigma}_{\m{\Omega}}^*}$ of (\ref{formu:GLS2_inc}), the parameter estimate $\sbra{\widehat{\m{f}}, \widehat{\m{p}}, \widehat{\m{\sigma}}_{\m{\Omega}}}$ can be obtained as well from the Vandermonde decomposition of $T\sbra{\m{u}^*}$ or $T\sbra{\m{u}^*}-\lambda_{\text{min}}\sbra{T\sbra{\m{u}^*}}\m{I}$.


To sum up, GLS is a gridless version of SPICE and can be applied to both the complete and missing data cases in the presence of either homoscedastic or heteroscedastic noise where the noise variance(s) is/are unknown. GLS carries out convex optimization in a reparameterized domain $\sbra{\m{u},\m{\sigma}}$ similarly to AST while the reparameterization process of GLS is done in a more explicit manner. The sparsity-promoting property of GLS is not so obvious as AST in which the atomic norm explicitly promotes sparsity. But it can be seen from (\ref{formu:GLS2_inc}) [and (\ref{formu:GLS2})] that the trace norm promotes sparsity while the second term plays data fitting. Since GLS requires neither the model order nor the noise variance, it might have some limitations in practice such as model order determination, which will be studied in the ensuing section. Motivated by that the GLS formulation in (\ref{formu:GLS2}) and the atomic norm in (\ref{formu:SDP_ANM}) are seemingly related, we will explore connections between GLS and atomic norm-based methods in Section \ref{sec:connection}.

\section{Extension of GLS: A Systematic Framework for Line Spectral Estimation} \label{sec:framework}

\subsection{Two Limitations of GLS: Inaccurate Model Order and Frequency Splitting} \label{sec:limits}
GLS may suffer from two limitations because it requires neither the model order nor the noise variance (and maybe more reasons). One is its inaccurate model order estimate, a common problem for sparse methods. GLS generally produces a frequency estimate of length much larger than the true model order. To see this, we consider the complete data case with homoscedastic noise as an example. According to Remark \ref{rem:Rtoeplitz}, the data covariance $\m{R}\geq\m{0}$ can be expressed as a Toeplitz matrix itself. By (\ref{formu:GLS2}) $\m{y}$ lies in the range space of $\m{R}$. Then the solution $\m{R}^*$ of $\m{R}$ has full rank with probability one, since otherwise $\m{y}\in\bC^M$ with random noise can be decomposed as superposition of $M-1$ sinusoids. Therefore, the number of estimated sinusoidal components of GLS, which equals the rank of $\m{R}^*-\lambda_{\text{min}}\sbra{\m{R}^*}\m{I}$, is almost sure to be $M-1$ (with probability zero the minimum eigenvalue is a multiple eigenvalue).

The other limitation is frequency splitting. An example is presented in Fig. \ref{Fig:frequencysplit}, where we attempt to estimate $K=3$ sinusoidal components using GLS from $L=50$ noisy samples that are randomly selected among $M=100$ measurements and corrupted by i.i.d. Gaussian noise with $\sigma\approx 0.4$. To rule out the possibility of numerical reasons, we solve GLS using a highly accurate SDP solver SDPT3 \cite{toh1999sdpt3} which is set to obtain the best precision. It is shown that the second component is split into two that are nearly located, which brings challenges to frequency estimation as well as model order selection.

\begin{figure}
\centering
  \includegraphics[width=2.8in]{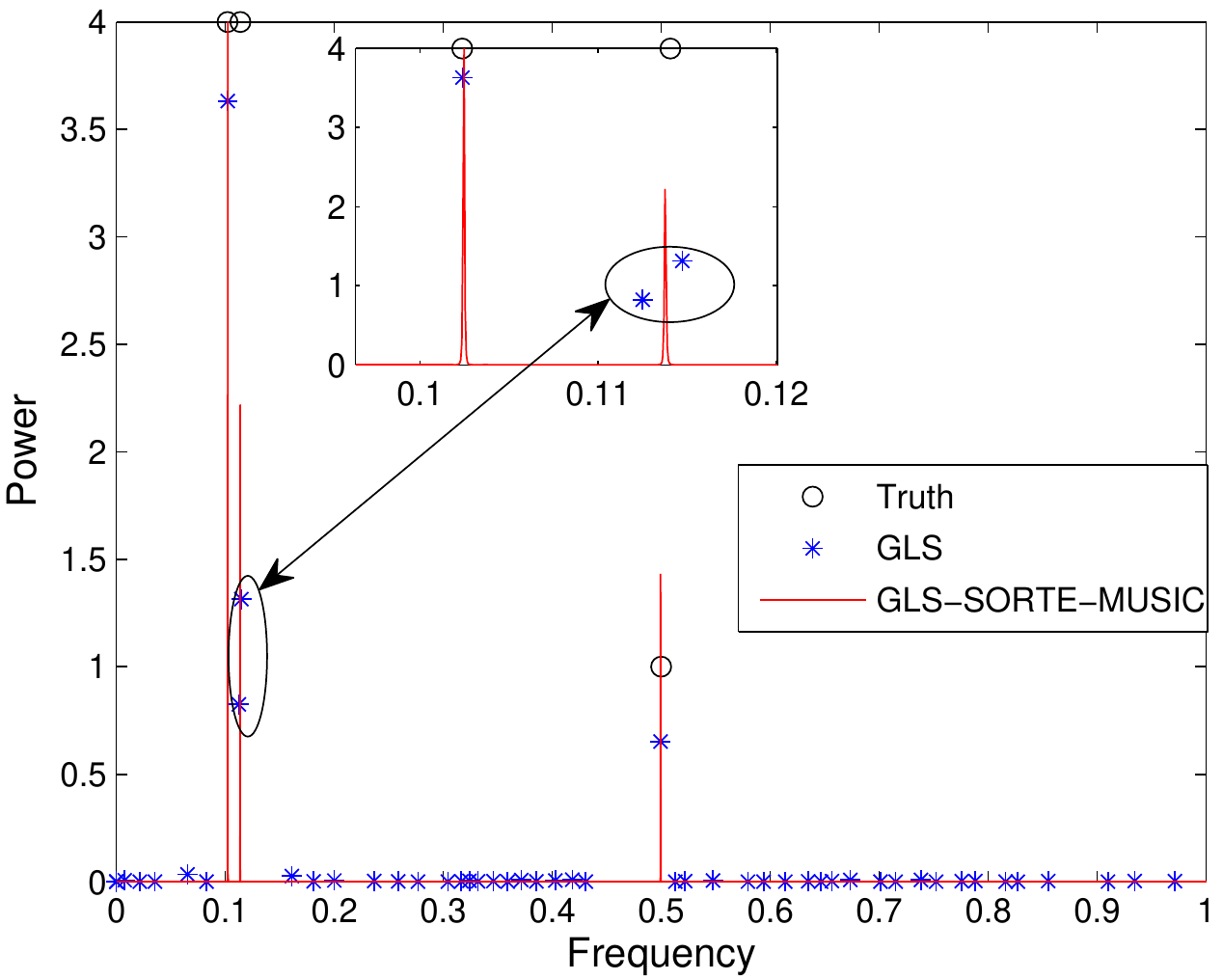}
\centering
\caption{Illustration of frequency splitting of GLS and its correction within the proposed framework. The three black circles indicate ground truth of the three sinusoidal components. Blue stars are produced by GLS using SDPT3, and the red curve of GLS-SORTE-MUSIC is given by GLS followed by SORTE for model order selection and MUSIC for spectral estimation. The area around the first two frequencies are zoomed in for better illustration.} \label{Fig:frequencysplit}
\end{figure}

\begin{rem} Note that the frequency splitting phenomenon reported in this paper is different from that encountered in grid-based methods (i.e., one sinusoidal component is split into a few supported on nearby grid points, see e.g., \cite{yang2012robustly,austin2013dynamic,bhaskar2013atomic}). In particular, the latter is caused by the grid and the convergence issue of certain algorithms \cite{austin2013dynamic}. In contrast, the frequency splitting shown in Fig. 1 is caused in part due to the absence of the noise level (note that AST rarely suffers from frequency splitting).
\end{rem}

\begin{rem} \label{rem:freqsplit} Since SPICE and $\ell_1$ optimization are approximate versions of GLS (more details will be shown in Section \ref{sec:connection}), they have the same frequency splitting problem (not due to the grid) without surprise, which has been confirmed in our simulations but we omit the details. This phenomenon has not been observed and reported in previous publications because of two reasons: 1) a rough grid leads to a worse frequency resolving resolution which can only detect frequency splitting caused by the grid, and 2) a highly dense grid means almost complete correlations between adjacent atoms and might result in numerical issues which bring challenges to the detection.
\end{rem}

\subsection{A Framework for Line Spectral Estimation} \label{sec:frameworkdetail}
To overcome the two limitations of GLS mentioned above, we propose the following framework for line spectral estimation which consists of three steps:
\begin{enumerate}
 \item Covariance estimation using GLS;
 \item Model order selection based on the covariance estimate;
 \item Frequency estimation based on the covariance and model order estimates.
\end{enumerate}
That is, we consider GLS as a covariance estimation scheme in this framework and then carry out line spectral estimation based on its solution. While these steps seem to be standard in covariance-based methods, the main contribution of this framework is the way that the covariance estimate is obtained. For example, it is generally not clear how to estimate the data covariance in the presence of missing data, while GLS provides a solution via covariance fitting and exploiting its Toeplitz structure. A data covariance estimate can be given in the complete data case by appropriately choosing a time window (see, e.g., \cite{stoica2005spectral}), however, the time window shortens the data length and potentially degrades the resolution limit \cite{candes2013towards}. In this paper, we choose the SORTE algorithm \cite{he2010detecting} for model order selection in \emph{Step 2} and MUSIC for the ensuing frequency estimation in \emph{Step 3}. It would be interesting to study in the future what choices of these methods result in the best performance.

In the framework, we estimate the model order from the covariance estimate $\m{\Gamma}_{\m{\Omega}}T\sbra{\m{u}^*}\m{\Gamma}_{\m{\Omega}}^T$ of the ``clean'' observed samples (to avoid effects of heteroscedastic noise). SORTE divides the eigenvalues of $\m{\Gamma}_{\m{\Omega}}T\sbra{\m{u}^*}\m{\Gamma}_{\m{\Omega}}^T$ into two clusters: one with larger eigenvalues corresponds to the signal subspace and the other the noise subspace. The estimated model order equals the size of the former cluster.
After that, frequency estimation is carried out using MUSIC. It is shown in Fig. \ref{Fig:frequencysplit} that the frequency splitting can be corrected within this framework, where the model order is correctly estimated.

\begin{rem} The proposed framework starts with covariance estimation and is applicable beyond GLS to grid-based methods such as SPICE and IAA. So, as a byproduct of the framework we also provide an off-grid frequency estimation approach to existing grid-based methods in the sense that the final frequency estimates are not constrained on the grid. Its effectiveness will be shown via numerical simulations. Moreover, it can also be applied to AST and $\ell_1$ optimization methods by connections of GLS and atomic norm based methods shown in the next section.
\end{rem}

\begin{rem} The proposed model order selection method is very different from conventional information-theoretic methods such as the minimum description length (MDL) principle, Akaike information criterion (AIC) and Bayesian information criterion (BIC) \cite{stoica2004model ,grunwald2007minimum}. For the aforementioned methods, one needs to solve a series of maximum likelihood estimation, or equivalently, NLS problems with respect to a set of candidate values of $K$, and then choose the best $K$. It is challenging to solve the NLS problems which require very accurate initialization since the objective functions have a \emph{complicated multimodal shape} with \emph{a very sharp global maximum} \cite{stoica2005spectral}. To date, there is no available method which is guaranteed to globally solve the NLS. The performance of existing initialization methods degrades in the presence of missing data. In contrast, the proposed method carries out convex optimization at the first step without the need of careful initialization.
\end{rem}


\section{Connections of GLS and Atomic Norm Denoising (AND)} \label{sec:connection}

\subsection{Basic Lemmas}
If the $\ell_1$ norm in L1ND [see (\ref{formu:l1ND})] is replaced by the atomic norm $\atomn{\cdot}$ (or $\atomni{\cdot}$), we call the resulting optimization problem
\equ{\min_{\m{z}} \mu \atomn{\m{z}}+g\sbra{\m{y}-\m{z}} \label{formu:AND}}
atomic norm denoising (AND). Correspondingly, (\ref{formu:AND}) with the three common choices of function $g\sbra{\cdot}$, including $\twon{\cdot}^2$, $\twon{\cdot}$ and $\onen{\cdot}$, is called gridless (GL-) Lasso, SR-Lasso and LAD-Lasso, respectively (GL-Lasso is exactly AST). To show connections of GLS and AND, we begin with some basic lemmas.

We use the following identity whenever $\m{R}\geq\m{0}$:
\equ{\m{y}^H\m{R}^{-1}\m{y}=\min_x x, \st \begin{bmatrix}x & \m{y}^H \\ \m{y} & \m{R}\end{bmatrix}\geq0. \label{formu:yRy}}
It follows that $\m{y}^H\m{R}^{-1}\m{y}$ is finite if and only if $\m{y}$ is in the range space of $\m{R}$. In fact, (\ref{formu:yRy}) is equivalent to defining $\m{y}^H\m{R}^{-1}\m{y}\triangleq\lim_{\sigma\rightarrow0_+} \m{y}^H\sbra{\m{R}+\sigma\m{I}}^{-1}\m{y}$ when $\m{R}$ loses rank. The following result follows from (\ref{formu:SDP_ANM}) and (\ref{formu:yRy}).

\begin{lem} It holds that
\equ{\atomn{\m{y}}=\min_{\m{u}} \frac{1}{2}u_1 +\frac{1}{2}\m{y}^H\mbra{T\sbra{\m{u}}}^{-1}\m{y}, \st T\sbra{\m{u}}\geq\m{0}. \label{formu:SDP_AN2}} \label{lem:AN} \BOX
\end{lem}

\begin{lem} Given $\m{R}=\m{A}\m{A}^H\geq\m{0}$, it holds that $\m{y}^H\m{R}^{-1}\m{y}=\min\twon{\m{x}}^2, \st \m{A}\m{x}=\m{y}$. \label{lem:lem1}
\end{lem}
\begin{proof} We need only to show that for any $\m{x}$ satisfying $\m{A}\m{x}=\m{y}$ it holds that $\m{y}^H\m{R}^{-1}\m{y}\leq\twon{\m{x}}^2$, or equivalently,
$\begin{bmatrix}\twon{\m{x}}^2 & \m{y}^H \\ \m{y} & \m{R}\end{bmatrix}\geq0$.
The conclusion follows from that $\begin{bmatrix}\twon{\m{x}}^2 & \m{y}^H \\ \m{y} & \m{R}\end{bmatrix} = \begin{bmatrix}\twon{\m{x}}^2 & \m{x}^H\m{A}^H \\ \m{A}\m{x} & \m{A}\m{A}^H\end{bmatrix}=\begin{bmatrix}\m{x}^H\\\m{A}\end{bmatrix}\begin{bmatrix}\m{x}^H\\\m{A}\end{bmatrix}^H\geq0$.
\end{proof}

\begin{lem} Given $\m{R}\geq\m{0}$ and $\m{\sigma}\succeq\m{0}$, it holds that
\equ{\begin{split}
&\m{y}^H\mbra{\m{R}+\diag\sbra{\m{\sigma}}}^{-1}\m{y}\\
&=\min_{\m{z}} \m{z}^H\m{R}^{-1}\m{z} + \sbra{\m{y}-\m{z}}^H\diag^{-1}\sbra{\m{\sigma}}\sbra{\m{y}-\m{z}}. \end{split}}
\label{lem:separate}
\end{lem}
\begin{proof} Since $\m{R}\geq\m{0}$ there exists a matrix $\m{A}$ satisfying that $\m{R}=\m{A}\m{A}^H$. It follows that $\m{R}+\diag\sbra{\m{\sigma}}=\begin{bmatrix} \m{A} & \diag^{\frac{1}{2}}\sbra{\m{\sigma}}\end{bmatrix} \begin{bmatrix} \m{A} & \diag^{\frac{1}{2}}\sbra{\m{\sigma}}\end{bmatrix}^H$. The following equalities hold by Lemma \ref{lem:lem1}:
\equ{\begin{split}
&\min_{\m{z}} \m{z}^H\m{R}^{-1}\m{z} + \sbra{\m{y}-\m{z}}^H\diag^{-1}\sbra{\m{\sigma}}\sbra{\m{y}-\m{z}}\\
=& \min_{\m{x},\m{A}\m{x}=\m{z}}\twon{\m{x}}^2 + \sbra{\m{y}-\m{z}}^H\diag^{-1}\sbra{\m{\sigma}}\sbra{\m{y}-\m{z}},\\
=& \min_{\m{x}}\twon{\m{x}}^2 + \sbra{\m{y}-\m{A}\m{x}}^H\diag^{-1}\sbra{\m{\sigma}}\sbra{\m{y}-\m{A}\m{x}}\\
=& \min_{\m{x},\m{d}}\twon{\m{x}}^2 + \twon{\m{d}}^2,\\
& \st \diag^{-\frac{1}{2}}\sbra{\m{\sigma}}\sbra{\m{y}-\m{A}\m{x}}=\m{d}\\
=& \min_{\m{x},\m{d}}\twon{\begin{bmatrix}\m{x}\\\m{d}\end{bmatrix}}^2, \st \begin{bmatrix} \m{A} & \diag^{\frac{1}{2}}\sbra{\m{\sigma}}\end{bmatrix}\begin{bmatrix}\m{x}\\\m{d}\end{bmatrix} =\m{y}\\
=& \m{y}^H\mbra{\m{R}+\diag\sbra{\m{\sigma}}}^{-1}\m{y}.
\end{split}}
\end{proof}

\begin{lem} Given $\m{R}\geq\m{0}$, $\min_{\m{y}_{\overline{\m{\Omega}}}}\m{y}^H\m{R}^{-1}\m{y} = \m{y}_{\m{\Omega}}^H\m{R}_{\m{\Omega}}^{-1}\m{y}_{\m{\Omega}}$. \label{lem:augment}
\end{lem}
\begin{proof} Suppose $\m{R}=\m{A}\m{A}^H$. It follows that $\m{R}_{\m{\Omega}}=\m{A}_{\m{\Omega}}\m{A}_{\m{\Omega}}^H$. By Lemma \ref{lem:lem1} it holds that
\equ{\begin{split}\min_{\m{y}_{\overline{\m{\Omega}}}}\m{y}^H\m{R}^{-1}\m{y}
&= \min_{\m{y}_{\overline{\m{\Omega}}}} \min_{\m{x}, \m{A}\m{x}=\m{y}}\twon{\m{x}}^2\\
&= \min_{\m{x}, \m{A}_{\m{\Omega}}\m{x}=\m{y}_{\m{\Omega}}} \twon{\m{x}}^2 = \m{y}_{\m{\Omega}}^H\m{R}_{\m{\Omega}}^{-1}\m{y}_{\m{\Omega}}. \end{split}}
\end{proof}

\subsection{Equivalence Between GLS and AND} \label{sec:equivalence}
We consider only the missing data case since the complete data case is a special case with $\m{\Omega}=\mbra{M}$. The result in the latter can be easily obtained by the substitutions $\m{\Omega}\rightarrow\mbra{M}$, $L\rightarrow M$, $\cA\sbra{\m{\Omega}}\rightarrow\cA$, $\m{y}_{\m{\Omega}}\rightarrow\m{y}$ and $\m{z}_{\m{\Omega}}\rightarrow\m{z}$. Since frequency estimation is of most importance which is determined by the solution of $\m{u}$ in both GLS and AND (formulated as SDPs), we use the following definition hereafter.
\begin{defi} We say that two optimization problems are equivalent if they produce the same solution $\m{u}^*$ up to a positive scale, i.e., they produce the same $\widehat{\m{p}}$ or $\abs{\widehat{\m{s}}}$ up to a positive scale and exactly the same frequency estimate $\widehat{\m{f}}$.
\end{defi}

\begin{thm} The GLS optimization problem in (\ref{formu:GLS2_inc}) is equivalent to one of the following AND problems:
\begin{enumerate}
 \item under the assumption of heteroscedastic noise,
  \equ{\min_{\m{z}_{\m{\Omega}}} \sqrt{L}\atomni{\m{z}_{\m{\Omega}}} +\onen{\m{y}_{\m{\Omega}}-\m{z}_{\m{\Omega}}}; \label{formu:equivl1_inc}}
 \item under homoscedastic noise,
  \equ{\min_{\m{z}_{\m{\Omega}}} \atomni{\m{z}_{\m{\Omega}}} +\twon{\m{y}_{\m{\Omega}}-\m{z}_{\m{\Omega}}}; \label{formu:equivl2_inc}}
 \item under homoscedastic noise with known variance $\sigma_0$,
  \equ{\min_{\m{z}_{\m{\Omega}}} \frac{\sqrt{L}}{\twon{\m{y}_{\m{\Omega}}}}\sigma_0\atomni{\m{z}_{\m{\Omega}}} +\frac{1}{2}\twon{\m{y}_{\m{\Omega}}-\m{z}_{\m{\Omega}}}^2. \label{formu:equivl2s_inc}}
\end{enumerate}
Moreover, the GLS optimization problems under different assumptions have the optimal solution $\frac{\twon{\m{y}_{\m{\Omega}}}}{\sqrt{L}}\m{u}^*$ given the solution $\m{u}^*$ of corresponding AND problems mentioned above (formulated as SDPs following from (\ref{formu:SDP_AN_inc})). \label{thm:equiv1_inc}
\end{thm}
\begin{proof} According to Lemma \ref{lem:augment}, (\ref{formu:GLS2_inc}) is equivalent to the following problem:
\equ{\min_{\m{u},\m{\sigma}_{\m{\Omega}}\succeq\m{0}, \m{y}_{\overline{\m{\Omega}}}} \tr\sbra{\m{R}_{\m{\Omega}}}+\twon{\m{y}_{\m{\Omega}}}^2 \m{y}^H\m{R}^{-1}\m{y},\\
\st T\sbra{\m{u}}\geq\m{0}, \label{formu:GLS3_inc}}
where $\m{R}$ and $\m{R}_{\m{\Omega}}$ are given in (\ref{formu:R}) and (\ref{formu:R_Omega}) respectively.

In \emph{Case 1}, the following equalities hold by consecutively applying Lemma \ref{lem:separate}, Lemma \ref{lem:AN} and Lemma \ref{thm:SDP_AN_inc} (the positive semidefinite constraint is omitted for brevity):
{\lentwo\equa{
&&\text{(\ref{formu:GLS3_inc}) }\notag\\
&&= \min_{\m{u},\m{\sigma}_{\m{\Omega}}\succeq\m{0},\m{y}_{\overline{\m{\Omega}}},\m{z} }Lu_1+\m{1}^T\m{\sigma}_{\m{\Omega}} + \twon{\m{y}_{\m{\Omega}}}^2 \m{z}^H\mbra{T\sbra{\m{u}}}^{-1}\m{z} \notag\\
&&\qquad +\twon{\m{y}_{\m{\Omega}}}^2\sbra{\m{y}-\m{z}}^H\diag^{-1}\sbra{\m{\sigma}}\sbra{\m{y}-\m{z}} \label{formu_equiv1}\\
&&= \min_{\m{u},\m{\sigma}_{\m{\Omega}}\succeq\m{0},\m{z} }2L\mbra{\frac{1}{2}u_1+ \frac{1}{2}\cdot\frac{\twon{\m{y}_{\m{\Omega}}}}{\sqrt{L}}\m{z}^H \mbra{T\sbra{
\m{u}}}^{-1} \frac{\twon{\m{y}_{\m{\Omega}}}}{\sqrt{L}}\m{z} } \notag\\
&&\qquad + \twon{\m{y}_{\m{\Omega}}}^2\sbra{\m{y}_{\m{\Omega}}-\m{z}_{\m{\Omega}}}^H\diag^{-1}\sbra{\m{\sigma}_{\m{\Omega}}} \sbra{\m{y}_{\m{\Omega}}-\m{z}_{\m{\Omega}}} \notag\\
&&\qquad + \m{1}^T\m{\sigma}_{\m{\Omega}} \label{formu_equiv2}\\
&&= \min_{\m{z}}2L\atomn{\frac{\twon{\m{y}_{\m{\Omega}}}}{\sqrt{L}}\m{z}} +2\twon{\m{y}_{\m{\Omega}}}\onen{\m{y}_{\m{\Omega}}-\m{z}_{\m{\Omega}}} \label{formu_equiv3}\\
&&= 2\twon{\m{y}_{\m{\Omega}}}\lbra{\min_{\m{z}_{\m{\Omega}}}\sqrt{L}\atomni{\m{z}_{\m{\Omega}}} +\onen{\m{y}_{\m{\Omega}}-\m{z}_{\m{\Omega}}}}. \label{formu_equiv4}
}}So the equivalence holds.

In \emph{Case 2}, the equalities (\ref{formu_equiv1}) and (\ref{formu_equiv2}) still hold. It then follows from $\m{1}^T\m{\sigma}_{\m{\Omega}}=L\sigma_1$ and $\diag\sbra{\m{\sigma}_{\m{\Omega}}}=\sigma_1\m{I}$ that
\equ{\begin{split}
&\text{(\ref{formu:GLS3_inc})}\\
&= \min_{\m{z}}2L\atomn{\frac{\twon{\m{y}_{\m{\Omega}}}}{\sqrt{L}}\m{z}} +2\sqrt{L}\twon{\m{y}_{\m{\Omega}}}\twon{\m{y}_{\m{\Omega}}-\m{z}_{\m{\Omega}}}\\
&= 2\sqrt{L}\twon{\m{y}_{\m{\Omega}}}\lbra{\min_{\m{z}_{\m{\Omega}}}\atomni{\m{z}_{\m{\Omega}}} +\twon{\m{y}_{\m{\Omega}}-\m{z}_{\m{\Omega}}}}. \end{split}\label{formu_equiv6}}
In \emph{Case 3} where the noise variance is fixed,
\equ{\begin{split} &\text{(\ref{formu:GLS3_inc})}\\
&= \min_{\m{z}}2L\atomn{\frac{\twon{\m{y}_{\m{\Omega}}}}{\sqrt{L}}\m{z}} +\frac{\twon{\m{y}_{\m{\Omega}}}^2}{\sigma_0}\twon{\m{y}_{\m{\Omega}}-\m{z}_{\m{\Omega}}}^2+L\sigma_0\\
&= \frac{2\twon{\m{y}_{\m{\Omega}}}^2}{\sigma_0} \lbra{\min_{\m{z}_{\m{\Omega}}}\frac{\sqrt{L}}{\twon{\m{y}_{\m{\Omega}}}}\sigma_0 \atomni{\m{z}_{\m{\Omega}}} +\frac{1}{2}\twon{\m{y}_{\m{\Omega}}-\m{z}_{\m{\Omega}}}^2}\\
&\quad+L\sigma_0. \end{split}\label{formu_equiv8}}
The problems in (\ref{formu:equivl1_inc}), (\ref{formu:equivl2_inc}) and (\ref{formu:equivl2s_inc}) are convex and can be formulated as SDPs following from (\ref{formu:SDP_AN_inc}). The relation between their optimal solutions of $\m{u}$ and the corresponding GLS formulations can be easily identified from (\ref{formu_equiv2}).
\end{proof}

\begin{rem}~ \label{rem:equivalence}
\begin{enumerate}
 \item Under homoscedastic noise, the SDP of GLS can be simplified according to Remark \ref{rem:Rtoeplitz}. With this simplified representation, the GLS optimization problem is equivalent to computing $\atomni{\m{y}_{\m{\Omega}}}$.
 \item In the limiting noiseless case, GLS is identical to the atomic norm method except that a postprocessing procedure is adopted in GLS for ensuring unique parameter estimate in the case where $T\sbra{\m{u}^*}$ has full rank.
\end{enumerate} \label{rem:equivalence}
\end{rem}

Theorem \ref{thm:equiv1_inc} shows that GLS can be interpreted as three different AND methods under different assumptions of noise. Under  heteroscedastic noise, GLS is in the form of GL-LAD-Lasso which tends to suppress significant noise entries and be robust to outliers by the use of the $\ell_1$ norm in data fitting. Under homoscedastic noise, GLS is in the form of GL-SR-Lasso in which all the noise entries are considered in whole and only the noise energy is reflected. Finally, under homoscedastic noise with fixed variance, GLS is in the form of GL-Lasso or AST and the noise variance is reflected in the regularization parameter.
Moreover, it is worth noting that, since SPICE and L1ND are grid-based versions of GLS and AND, similar equivalence exists between them. Part of the result has been shown in \cite{rojas2013note,babu2014connection}.

\begin{rem} We can now formally show that GLS is equivalent to and provides a practical implementation of the limiting scenario of SPICE. Denote by $\text{SPICE}_N^*$ and $\text{GLS}^*$ the optimal objective function values of SPICE (with a uniform grid of size $N$) and GLS, respectively. By inserting (\ref{formu:Natomnvsatomn}) into the optimization problems in Theorem \ref{thm:equiv1_inc} we obtain that
\equ{\sbra{1-\frac{\pi \overline{M}}{N}}\text{SPICE}_N^* \leq \text{GLS}^* \leq \text{SPICE}_N^*. }
It therefore holds that $\lim_{N\rightarrow +\infty}\text{SPICE}_N^*= \text{GLS}^*$.
\end{rem}

\subsection{Implications of the Equivalence}
\subsubsection{Overfitting under homoscedastic noise} Under homoscedastic noise, the GLS optimization problem is equivalent to computing $\atomni{\m{y}_{\m{\Omega}}}$ by Remark \ref{rem:equivalence}. That is, GLS carries out the optimization as in the noiseless case and results in overfitting, indicating necessity of model order selection and modified frequency estimation discussed in Section \ref{sec:framework}. Note that the parameter estimation process of GLS given in Section \ref{sec:GLS} slightly alleviates the overfitting problem.

\subsubsection{Suboptimal power estimation}  \label{sec:subopt}
By comparing the frequency retrieval processes of GLS and its equivalent AND formulations in Theorem \ref{thm:equiv1_inc} (see Sections \ref{sec:GLS_com} and \ref{sec:AST} respectively), one may find that the power estimate $\widehat{\m{p}}$ of GLS is inherently the amplitude $\abs{\widehat{\m{s}}}$ (scaled by $\frac{\twon{\m{y}_{\m{\Omega}}}}{\sqrt{L}}$) rather than its square, the power. Similarly, $\widehat{\m{\sigma}}_{\m{\Omega}}$ of GLS estimates standard deviation of the noise (scaled by $\twon{\m{y}_{\m{\Omega}}}$) by derivations of (\ref{formu_equiv3}) and (\ref{formu_equiv6}) instead of the variance. Similar arguments have been made in \cite{stoica2011new,stoica2011spice}. These claims are verified by numerical simulations in Section \ref{sec:simulation}.


\begin{rem} Utilizing the aforementioned interpretation of $\widehat{\sigma}$ (when $\sigma$ is estimated from the data), the regularization constant $\frac{\sqrt{L}}{\twon{\m{y}_{\m{\Omega}}}}\sigma_0$ in (\ref{formu:equivl2s_inc}) will be modified into $\frac{\sqrt{L}}{\twon{\m{y}_{\m{\Omega}}}}\times \twon{\m{y}_{\m{\Omega}}}\sigma_0^{\frac{1}{2}}=\sqrt{L}\sigma_0^{\frac{1}{2}}$. It is interesting to note that this constant, without any distribution assumed for the noise, is close to the optimized value $\approx\sqrt{L\ln \overline{M}}\sigma_0^{\frac{1}{2}}$ given in Theorem \ref{thm:AST_inc} for Gaussian noise.
\end{rem}


\section{Computationally Feasible Solutions} \label{sec:ADMM}

\subsection{Exact Solutions via Duality}
We solve GLS in the elegant AND forms in Theorem \ref{thm:equiv1_inc}. Using the standard SDP solver SDPT3 \cite{toh1999sdpt3}, we empirically find that faster speed can be achieved by solving their dual problems. Meanwhile, the solutions of the primal problems are given for free. Consider GL-LAD-Lasso in (\ref{formu:equivl1_inc}) as an example. By Lemma \ref{thm:SDP_AN_inc} it can be written into the following SDP:
\equ{\min_{x,\m{u}, \m{z}} \tau\sbra{x+ u_1} + \onen{\m{y}_{\m{\Omega}} - \m{z}_{\m{\Omega}}}, \st \begin{bmatrix}x& \m{z}^H \\ \m{z} & T\sbra{\m{u}} \end{bmatrix}\geq\m{0}, \label{formu:SDP_AN_inc2}}
where $\tau=\frac{\sqrt{L}}{2}$. Its dual problem is given by the following SDP following a standard Lagrangian analysis \cite{boyd2004convex}:
\equ{\min_{\m{v},\m{W}} 2\Re\lbra{\m{y}_{\m{\Omega}}^H \m{v}_{\m{\Omega}}}, \st \left\{\begin{array}{l}\begin{bmatrix}\tau & \m{v}^H \\ \m{v} & \m{W} \end{bmatrix}\geq\m{0},\\ \m{v}_{\overline{\m{\Omega}}}=\m{0},\\ \inftyn{\m{v}_{\m{\Omega}}}\leq\frac{1}{2},\\ T^*\sbra{\m{W}}=\tau\m{e}_1,\end{array}\right. \label{formu:dualGLLADLasso}}
where $T^*\sbra{\cdot}$ denotes the adjoint operator of $T\sbra{\cdot}$ and $\m{e}_1=\mbra{1,0,\dots,0}^T\in\bR^M$.

\subsection{Exact Solutions via ADMM}
SDPT3 implements the interior point method and does not scale well with the problem dimension. In this subsection we present a first-order algorithm for the SDPs involved in this paper based on ADMM which is a well-established method for large scale problems \cite{boyd2011distributed}. We provide only the algorithm for GL-LAD-Lasso in (\ref{formu:SDP_AN_inc2}) for brevity while those for AND and GL-SR-Lasso can be derived similarly. (\ref{formu:SDP_AN_inc2}) can be written into the following form:
\equ{\begin{split}
&\min_{x,\m{u}, \m{z}, \m{Q}\geq\m{0}} \tau\sbra{x+ u_1} + \onen{\m{y}_{\m{\Omega}} - \m{z}_{\m{\Omega}}},\\
&\quad\st \m{Q} = \begin{bmatrix}x& \m{z}^H \\ \m{z} & T\sbra{\m{u}} \end{bmatrix}. \end{split} \label{formu:SDP_AN_inc3}}
We introduce $\m{\Lambda}$ as the Lagrangian multiplier. Then the augmented Lagrange function of (\ref{formu:SDP_AN_inc3}) is
\equ{\begin{split}
&\cL_A\sbra{x,\m{u}, \m{z}, \m{Q}, \m{\Lambda}} \\
=& \tau\sbra{x+ u_1} + \onen{\m{y}_{\m{\Omega}} - \m{z}_{\m{\Omega}}} + \tr\mbra{\sbra{ \m{Q} - \begin{bmatrix}x& \m{z}^H \\ \m{z} & T\sbra{\m{u}} \end{bmatrix} }\m{\Lambda}}\\
&+ \frac{\beta}{2}\frobn{\m{Q} - \begin{bmatrix}x& \m{z}^H \\ \m{z} & T\sbra{\m{u}} \end{bmatrix}}^2,
\end{split}}
where $\beta>0$ is a penalty parameter set according to \cite{boyd2011distributed}.
Following the routine of ADMM, the variables $\sbra{x,\m{u}, \m{z}, \m{Q}}$ and $\m{\Lambda}$ can be iteratively updated in close forms. We omit the details due to the page limit. The resulting algorithm converges to the optimal solution of (\ref{formu:SDP_AN_inc2}) by \cite{boyd2011distributed}.

\subsection{Approximate Solutions via Frequency Discretization} \label{sec:discretization} 
The ADMM algorithm is more scalable with the problem dimension compared to SDPT3, however, it still needs to carry out one eigen-decomposition of a matrix of order $M+1$ at each iteration which is computationally expensive when $M$ is large. Since we have shown previously that the grid-based $\ell_1$ techniques (including SPICE) are good approximations of AST and GLS if the grid size $N=O\sbra{M}$, they are reasonable substitutions for which many computationally efficient algorithms have been developed such as SPGL1 and ONE-L1 \cite{van2008probing,yang2011orthonormal}. But we note that the $\ell_1$ techniques might suffer from some shortcomings for line spectral estimation with a highly dense grid: 1) almost complete correlations between adjacent atoms typically cause computational issues such as slow convergence and even computational instability, and 2) basis mismatches cause frequency splitting, resulting in underestimation of the power and less accurate frequency estimate.


\section{Numerical Simulations} \label{sec:simulation}

\subsection{Equivalence Between GLS and AND}
We first verify the equivalence between GLS and AND. Suppose that we observe $L=30$ randomly located samples of $M=50$ consecutive measurements of a sinusoidal signal composed of $K=3$ sinusoids with true parameters $\m{f}=\mbra{0.1, 0.12, 0.5}^T$ and $\m{p}=\mbra{9,4,1}^T$. The measurements are contaminated with i.i.d. Gaussian noise with $\sigma=0.1$. The first two frequencies are separated by $\frac{1}{M}=0.02$. We carry out line spectral estimation using both GLS and its equivalent AND forms in Theorem \ref{thm:equiv1_inc} under the assumption of heteroscedastic or homoscedastic noise. We plot in Fig. \ref{Fig:equiv} power estimates of GLS scaled by $\frac{\sqrt{L}}{\twon{\m{y}_{\m{\Omega}}}}$ and amplitude estimates of AND versus their frequency estimates, respectively. As concluded in Theorem \ref{thm:equiv1_inc}, we see that under both the assumptions they obtain identical results. The results under the two assumptions differ slightly from each other. Motivated by this observation, we consider only the assumption of heteroscedastic noise in the rest simulations for clearer presentation.

\begin{figure}
\centering
  \includegraphics[width=3in]{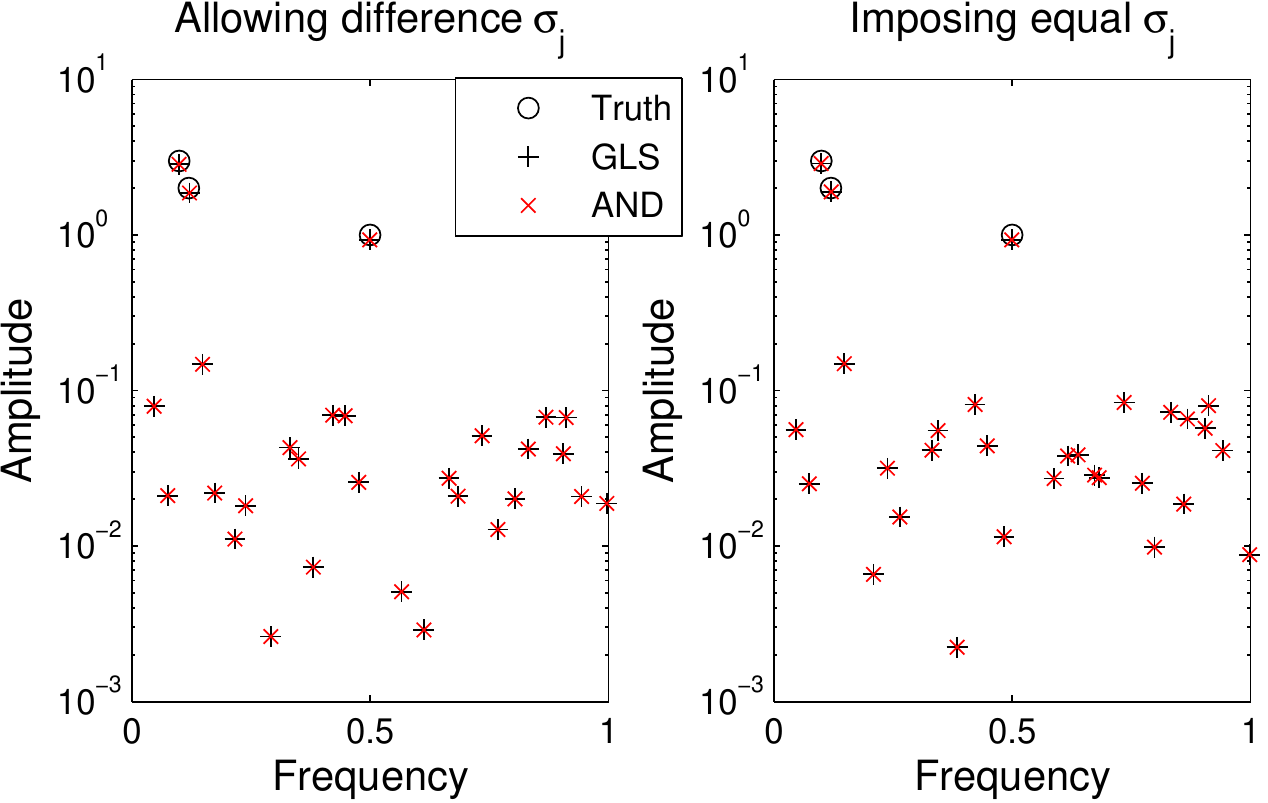}
\centering
\caption{Equivalence between GLS and AND under the assumption of (left) heteroscedastic or (right) homoscedastic noise. The amplitude estimates of GLS are scaled by $\frac{\sqrt{L}}{\twon{\m{y}_{\m{\Omega}}}}$ for better illustration.} \label{Fig:equiv}
\end{figure}

\subsection{Spectra Comparison of Gridless and Grid-based Methods}
We compare spectra of the gridless methods presented in this paper and the grid-based $\ell_1$ and SPICE methods. In our simulation, we set $M=100$, $L=50$, $K=3$, the true frequencies $\m{f}=\mbra{0.103, 0.115, 0.5}^T$ and corresponding powers $\m{p}=\mbra{4,4,1}^T$, and variance of i.i.d. Gaussian noise $\sigma=1$. AST and GLS are implemented using both SDPT3 and ADMM introduced in Section \ref{sec:ADMM}.  We consider two implementations of SPICE: one as in \cite{stoica2011new} and the other by solving its equivalent L1ND version implemented by SPGL1 \cite{van2008probing}. SPICE in \cite{stoica2011new} is considered converged if the objective function value decreases relatively by less than $10^{-6}$ between two consecutive iterations, or the maximum number of iterations, set to 2000, is reached. The Matlab code of SPGL1 is downloaded at http://www.cs.ubc.ca/$\sim$mpf/spgl1, and we use all default parameter settings but $\text{decTol} = 10^{-6}$ and maximum number of iterations 10000. Moreover, we consider two discretization levels for SPICE with $N=5M$ and $N=10M$. Note that the amplitude estimate of the SPICE algorithm in \cite{stoica2011new} suffers from a constant-factor ambiguity. The first two frequencies are located off the grid in SPICE with $N=5M$ and on the grid with $N=10M$ while the last frequency lies on the grid in both the cases.

Fig. \ref{Fig:comp_spectra} presents our simulation results over 5 Monte Carlo runs, where the spectra of AST and GLS solved by SDPT3 are omitted since they are visually identical to those by ADMM. All of GLS and its grid-based versions correctly detect the three components with some small spurious ones, while AST removes almost all of the spurious components with the oracle noise variance. By taking into account the constant-factor ambiguity of SPICE, the results of SPICE and SPGL1 differ slightly from each other. By comparing (b), (d) and (f), we see that SPICE tends to underestimate the power due to basis mismatches and slow convergence of SPGL1, which becomes more significant in the presence of off-grid frequencies and/or a denser grid. In computational time, AST and GLS with SDPT3 take 5.19s and 4.92s on average. The presented six methods (a)-(f) take 0.97s, 2.21s, 0.11s, 2.29s, 0.23s and 3.99s, respectively. We see that SPICE is the fastest in this example. The ADMM-based gridless methods introduced in this paper are faster than SPGL1 which converges slowly due to almost complete correlations between adjacent atoms.


\begin{figure*}
\centering
  \subfigure[AST by ADMM]{
    \label{Fig:spectrum_AST-ADMM}
    \includegraphics[width=2.3in]{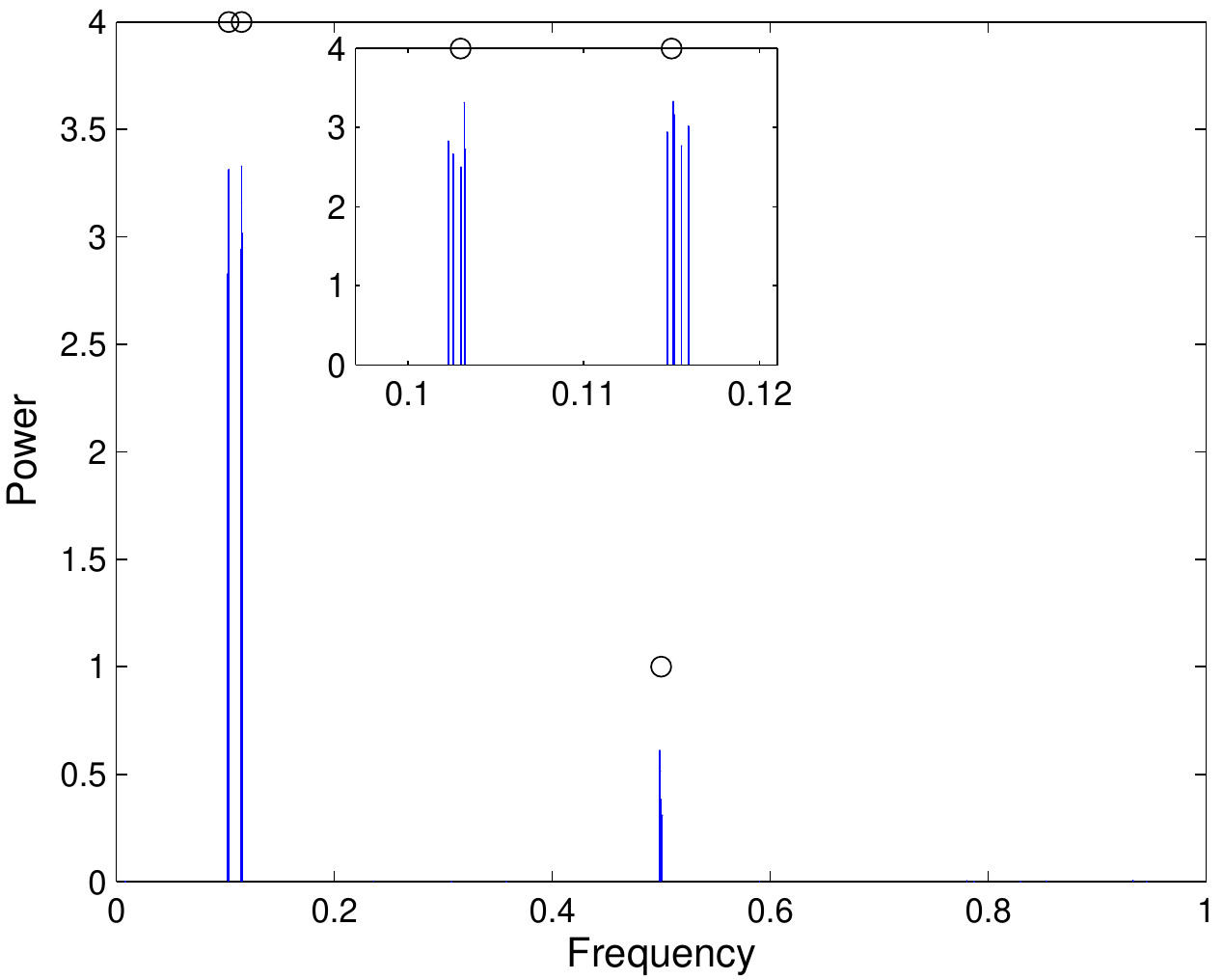}} %
  \subfigure[GLS by ADMM]{
    \label{Fig:spectrum_GLS-ADMM}
    \includegraphics[width=2.3in]{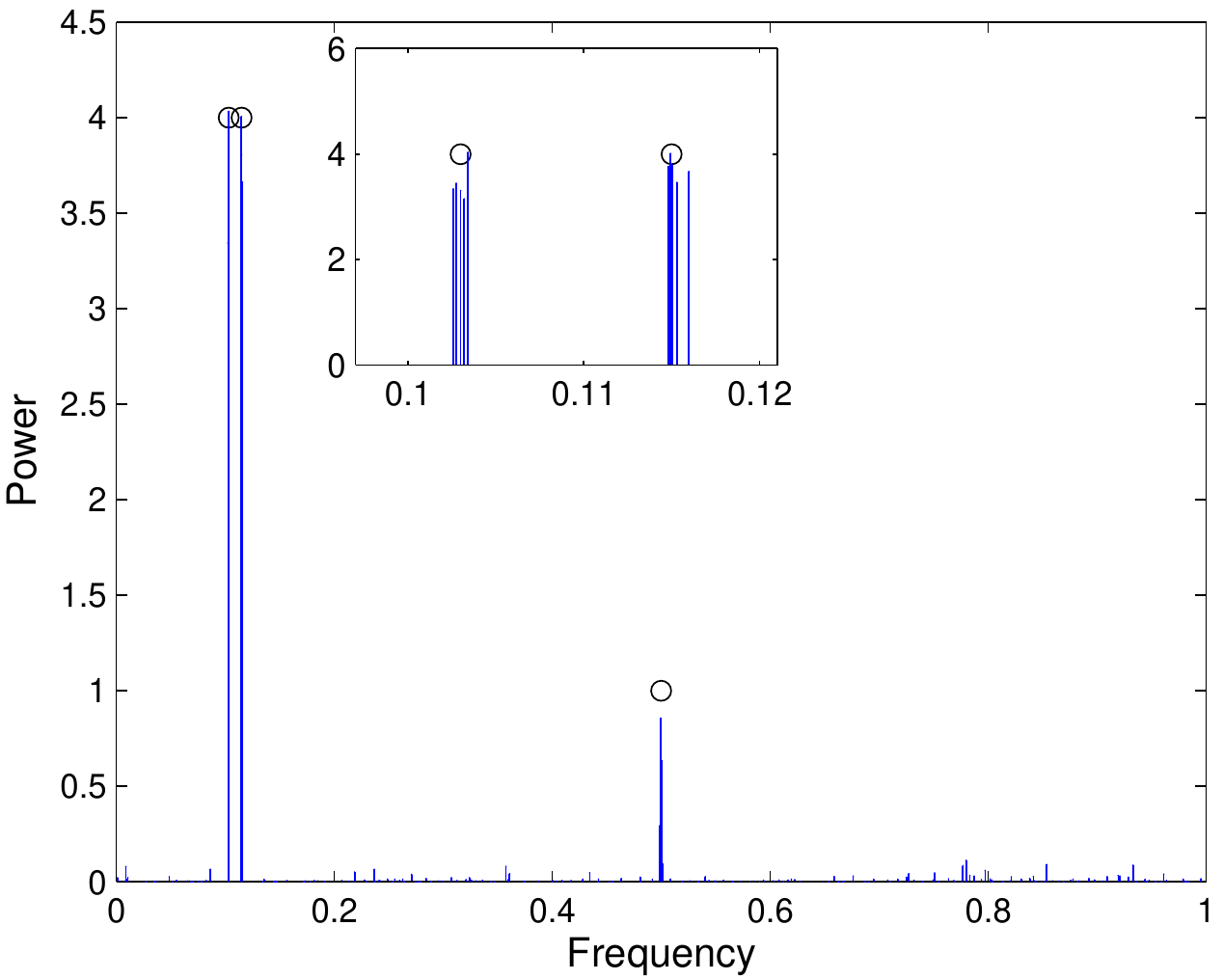}} %
  \subfigure[SPICE \cite{stoica2011new}, $N=5M$]{
    \label{Fig:spectrum_SPICE1}
    \includegraphics[width=2.3in]{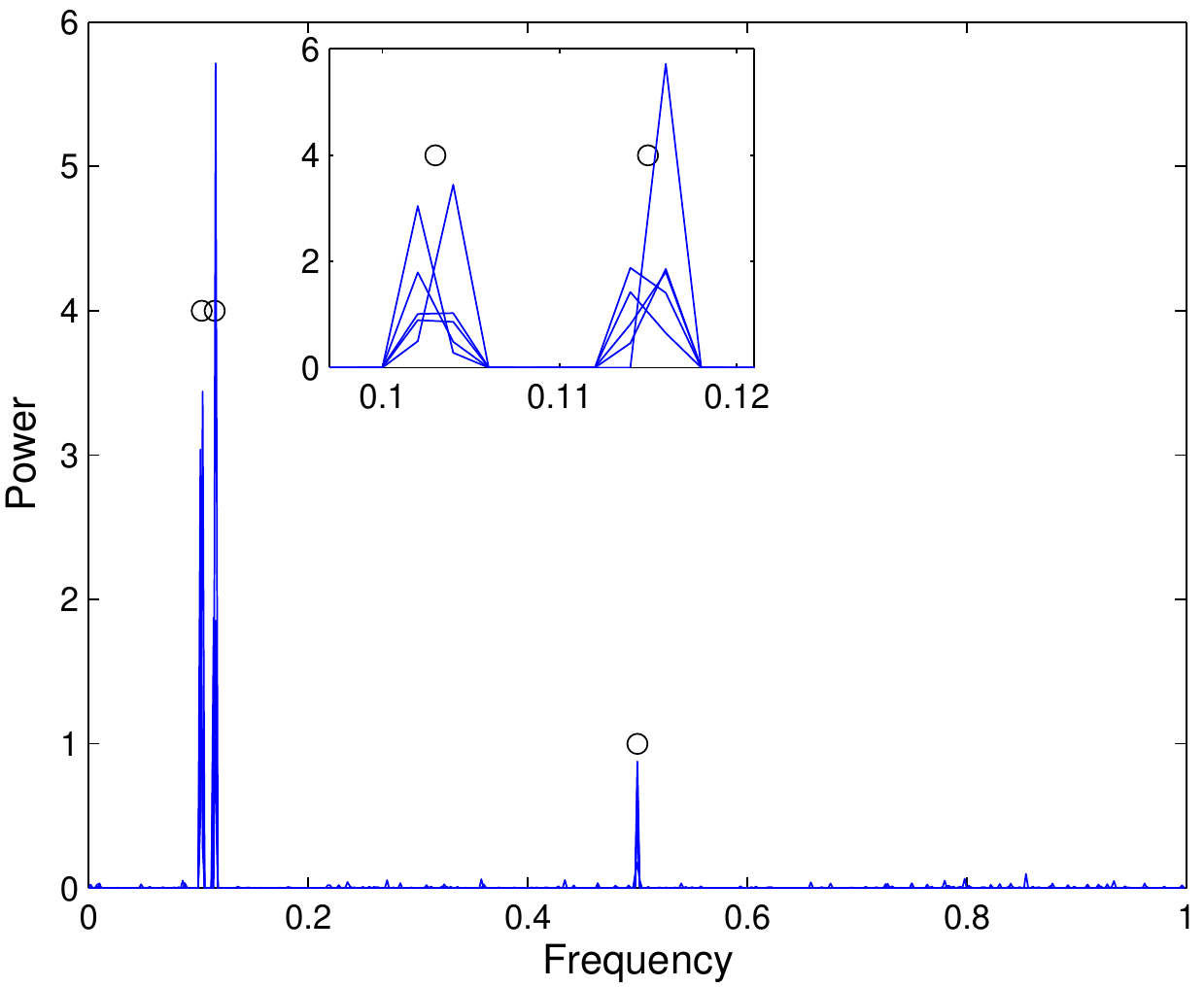}}
  \subfigure[SPICE by SPGL1, $N=5M$]{
    \label{Fig:spectrum_spg1}
    \includegraphics[width=2.3in]{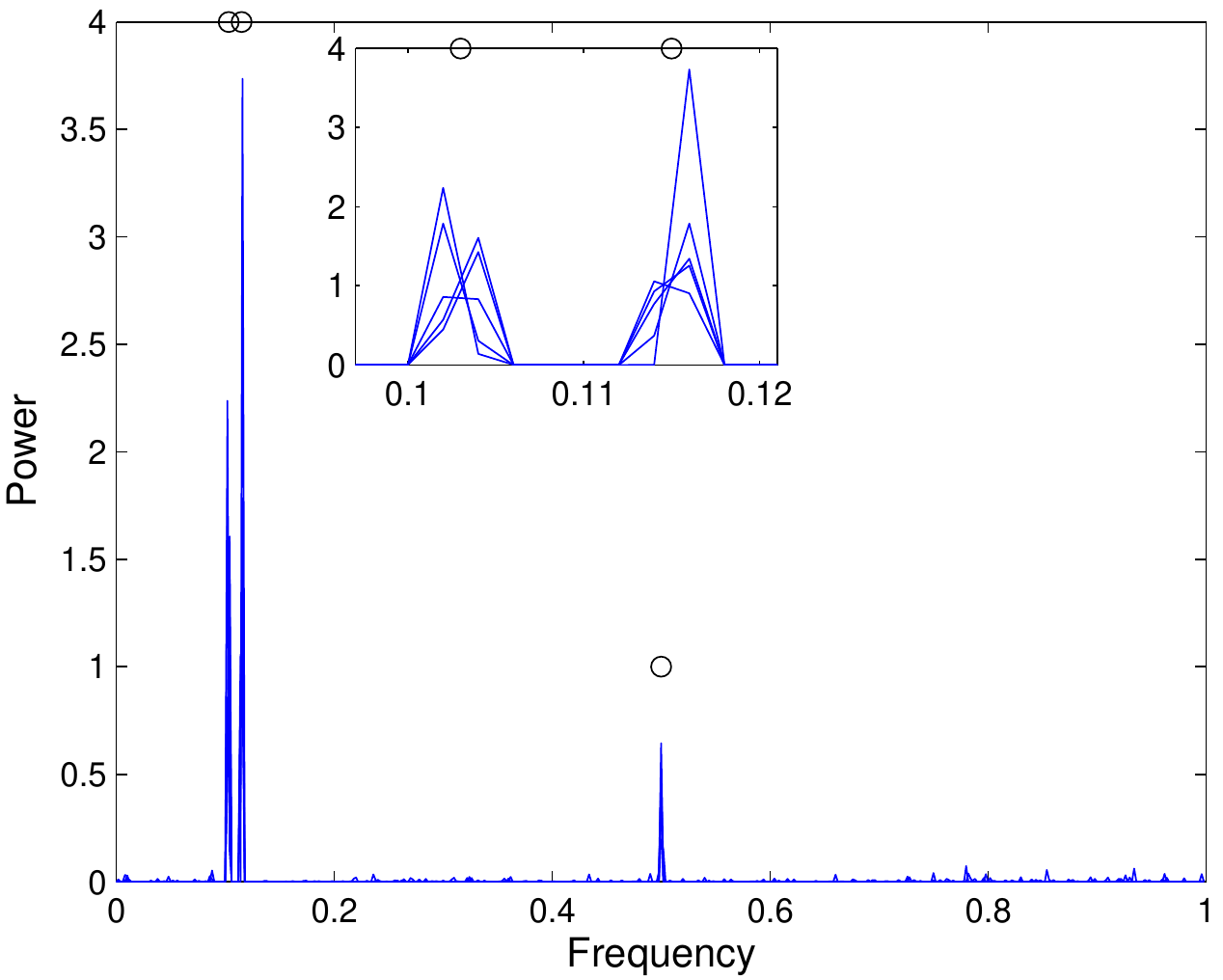}}%
  \subfigure[SPICE \cite{stoica2011new}, $N=10M$]{
    \label{Fig:spectrum_SPICE2}
    \includegraphics[width=2.3in]{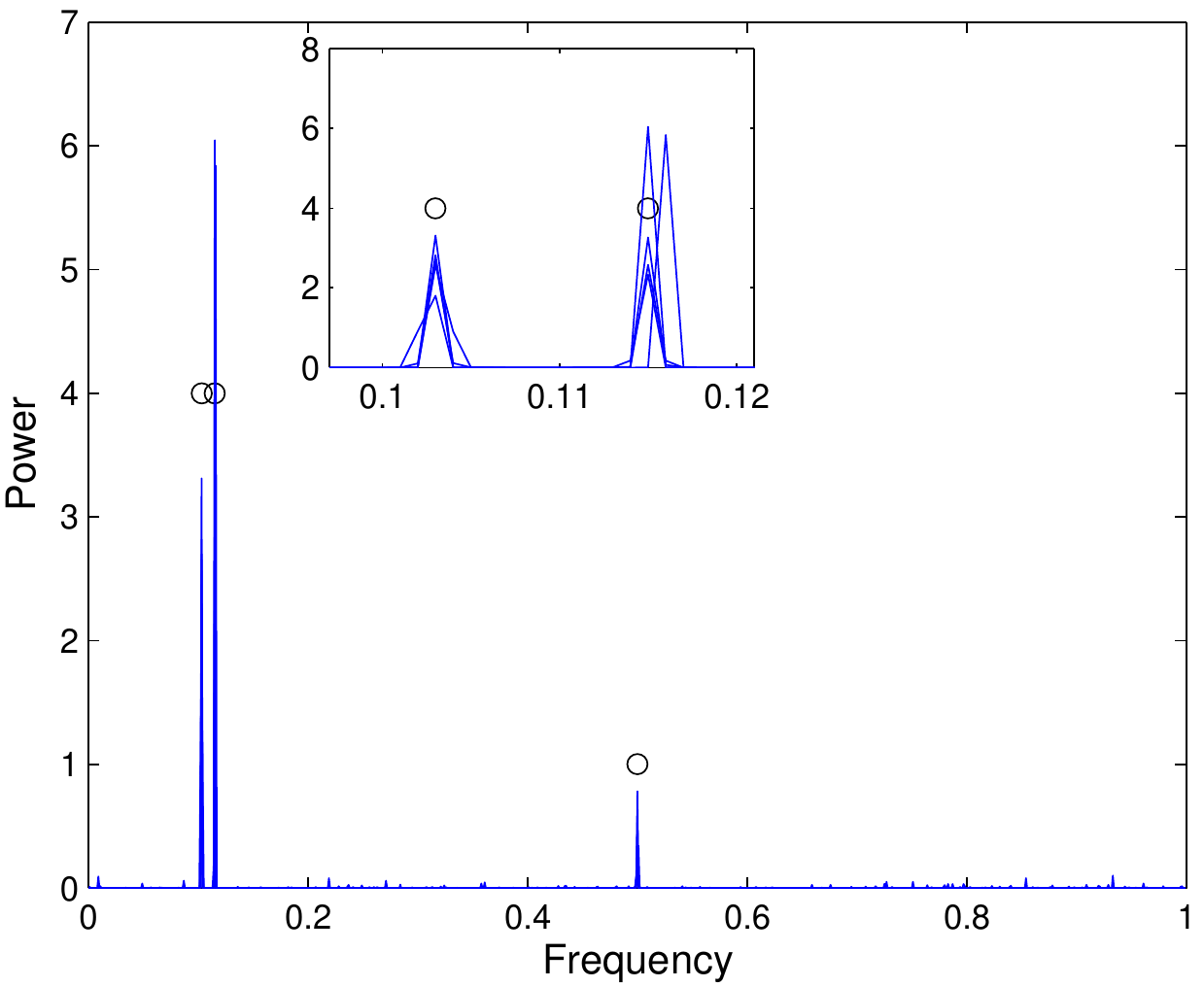}}%
  \subfigure[SPICE by SPGL1, $N=10M$]{
    \label{Fig:spectrum_spg2}
    \includegraphics[width=2.3in]{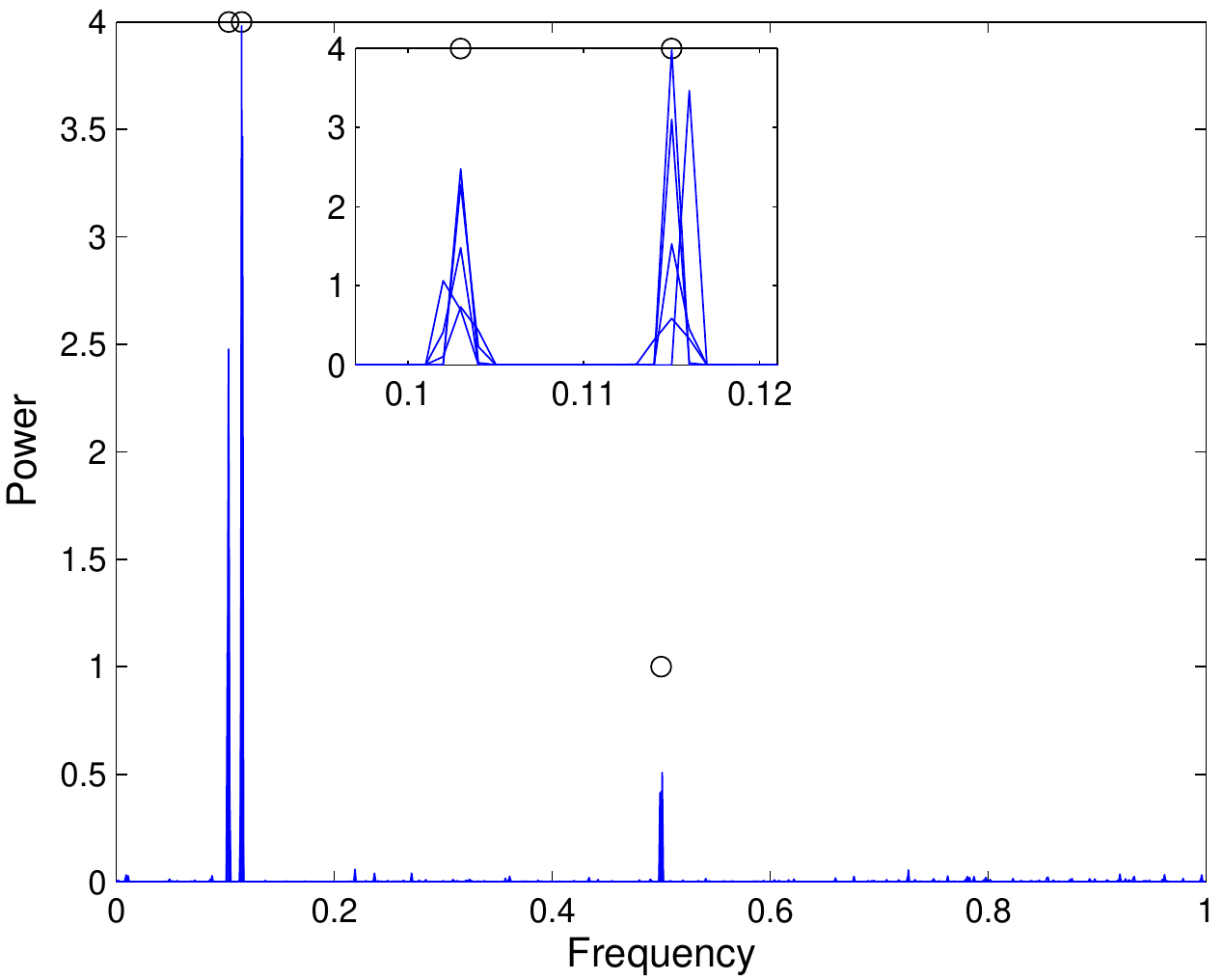}}
\centering
\caption{Power spectra of six methods/algorithms for line spectral estimation over 5 random Monte Carlo runs. Black circles indicate the ground truth with frequencies $\m{f}=\mbra{0.103, 0.115, 0.5}^T$ and powers $\m{p}=\mbra{4,4,1}^T$. The area around the first two frequency components is zoomed in in each subfigure. Other settings include $M=100$, $L=50$, and variance of i.i.d. Gaussian noise $\sigma=1$.} \label{Fig:comp_spectra}
\end{figure*}

\subsection{Model Order Selection and Frequency Estimation} \label{sec:simu_modelorder}
We examine performance of GLS, compared to its grid-based versions, in model order selection and frequency estimation in this subsection. In our simulation, we set $M=100$, $L=50$, $K=3$ and $\m{p}=\mbra{4,4,1}^T$ as before but generate the $K$ frequencies uniformly at random in intervals $(0.102,0.104)$, $(0.114, 0.116)$ and $(0.499, 0.501)$, respectively. Moreover, we define the signal-to-noise ratio $\text{SNR}=-10\log_{10}\sigma$ (i.e., with respect to the smallest component), and consider values of SNR from $-20$ to $20$dB at a step of $2$dB. A number of $100$ Monte Carlo runs are carried out at each value. We quantify the accuracy of frequency estimation of GLS, SPGL1 and SPICE in terms of MSE without and within the framework proposed in Section \ref{sec:framework}, respectively. In the former case, the frequency estimate of SPICE (and SPGL1) is given by the highest $K$ peaks to calculate the MSE while it corresponds to the largest $K$ components for GLS (outliers can be caused to GLS due to frequency splitting as illustrated in Fig. \ref{Fig:frequencysplit}). In the latter case, we use the Matlab routine \texttt{rootmusic} for frequency estimation. To quantify the performances of model order selection and frequency estimation independently, we use the exact model order $K$ in \texttt{rootmusic} rather than the one given by SORTE. More rationales behind this setting will be clarified later.


Our estimation results of the model order (using SORTE) and the $K$ frequencies are presented in Fig. \ref{Fig:varySNR}. Model order selection is considered successful if the estimated model order equals the true value. Fig. \ref{Fig:varySNR_modelorder} shows that the performance of GLS is very convincing in model order selection as the SNR is 0dB or above. In fact, only 2 failures occur out of all the 1100 trials when $\text{SNR}\geq0$dB (a careful study shows that both the failures overestimate the order by 1). SPICE and SPGL1 have similar performances when $N=10M$. When $N=5M$, however, the model order selection becomes less accurate in the high SNR regime, where the approximation errors of grid-based methods become non-negligible compared to noise. A careful recheck shows that SPICE and SPGL1 tend to overestimate the model order by 1.

The MSEs of frequency estimates, which are calculated using the true model order, are presented in Figs. \ref{Fig:varySNR_MSEf_N5M} and \ref{Fig:varySNR_MSEf_N10M}. Each MSE curve (except GLS) is divided into two parts at some critical SNR value, above which all the frequencies are accurately estimated. For SPICE and SPGL1, the MSEs without the framework are lower bounded by $\frac{1}{12N^2}$ (the horizontal black dashed lines) in expectation since the best frequency estimate is the nearest grid point \cite{yang2013off}. In contrast, GLS can outperform the lower bound but can be subject to 1 or 2 outliers (out of 100 runs) caused by frequency splitting as discussed in Section \ref{sec:limits} (see the points of discontinuity in the GLS MSE curve in Fig. \ref{Fig:varySNR_MSEf_N5M}). Within the proposed framework, GLS can stably estimate all the frequencies. Moreover, SPICE and SPGL1 can also outperform the aforementioned lower bound within the framework. Remarkably, their MSE curves coincide with that of GLS for $N=5M$. The results are similar for $N=10M$, however, convergence issues arise with this dense grid which, as shown in Fig. \ref{Fig:varySNR_MSEf_N10M}, cause an outlier to SPICE at $\text{SNR}=0$dB and worse performance to SPGL1-MUSIC in the high SNR regime.

We have used the exact model order $K$ in frequency estimation. In fact, it is shown in Fig. \ref{Fig:varySNR_modelorder} that the model order can be accurately estimated within the framework in the presence of modest or light noise (e.g., when $\text{SNR}\geq 0$dB). Otherwise, in the presence of heavy noise (e.g., when $\text{SNR}<-2$dB), it is shown in Figs. \ref{Fig:varySNR_MSEf_N5M} and \ref{Fig:varySNR_MSEf_N10M} that the frequencies cannot be reliably estimated even with the oracle $K$.

\begin{figure*}
\centering
  \subfigure[Model order selection]{
    \label{Fig:varySNR_modelorder}
    \includegraphics[width=2.3in]{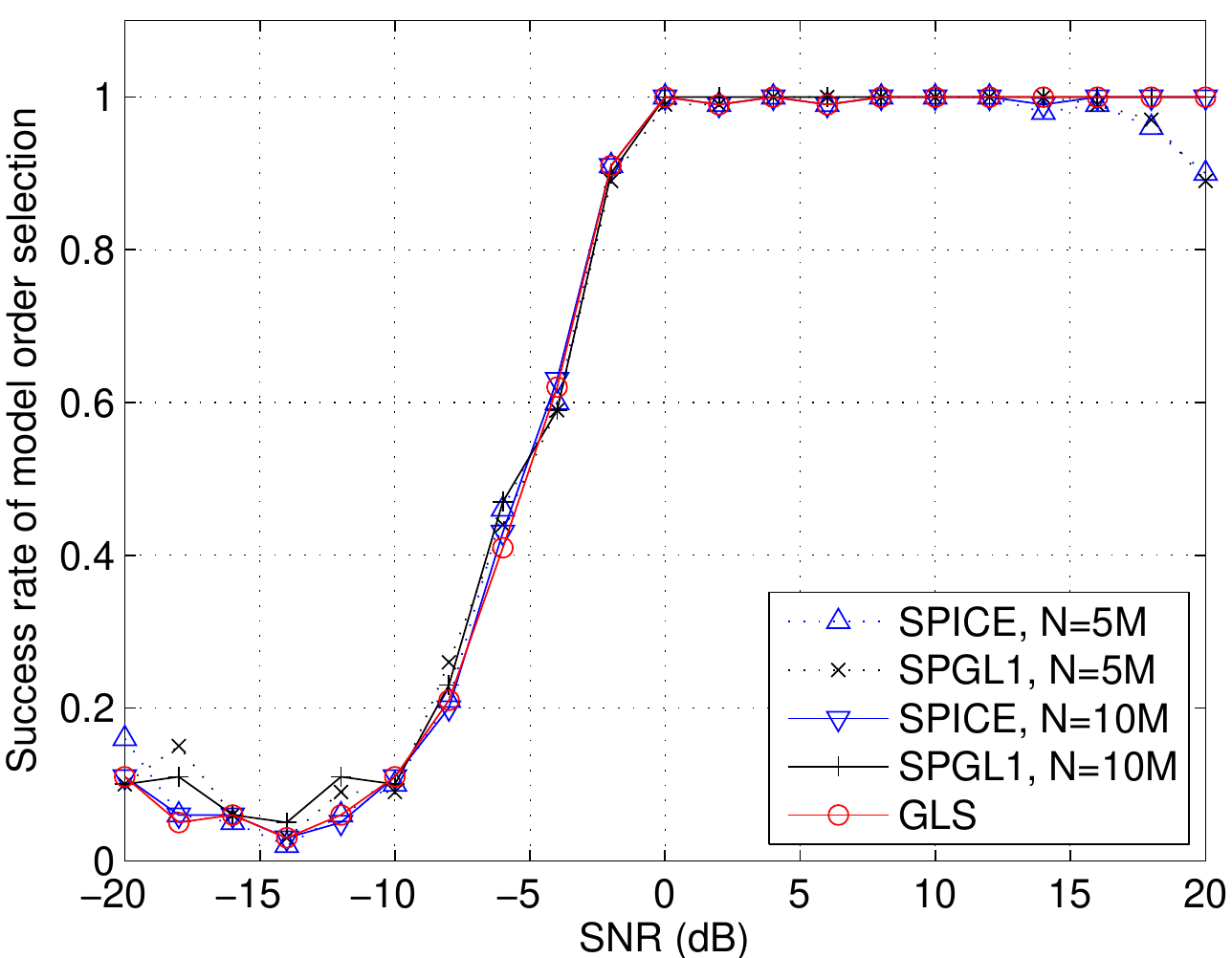}} %
  \subfigure[Frequency estimation, $N=5M$]{
    \label{Fig:varySNR_MSEf_N5M}
    \includegraphics[width=2.3in]{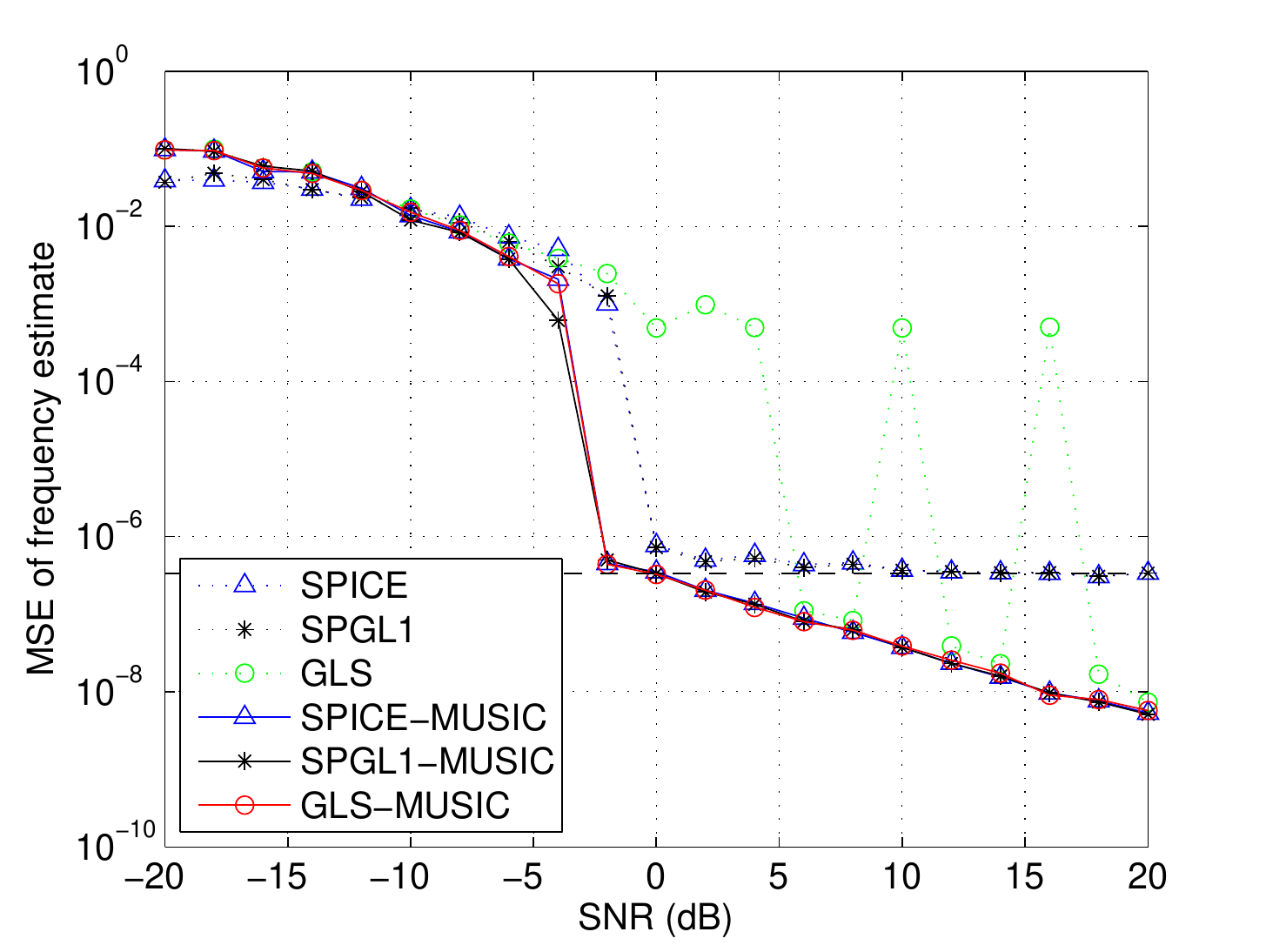}} %
  \subfigure[Frequency estimation, $N=10M$]{
    \label{Fig:varySNR_MSEf_N10M}
    \includegraphics[width=2.3in]{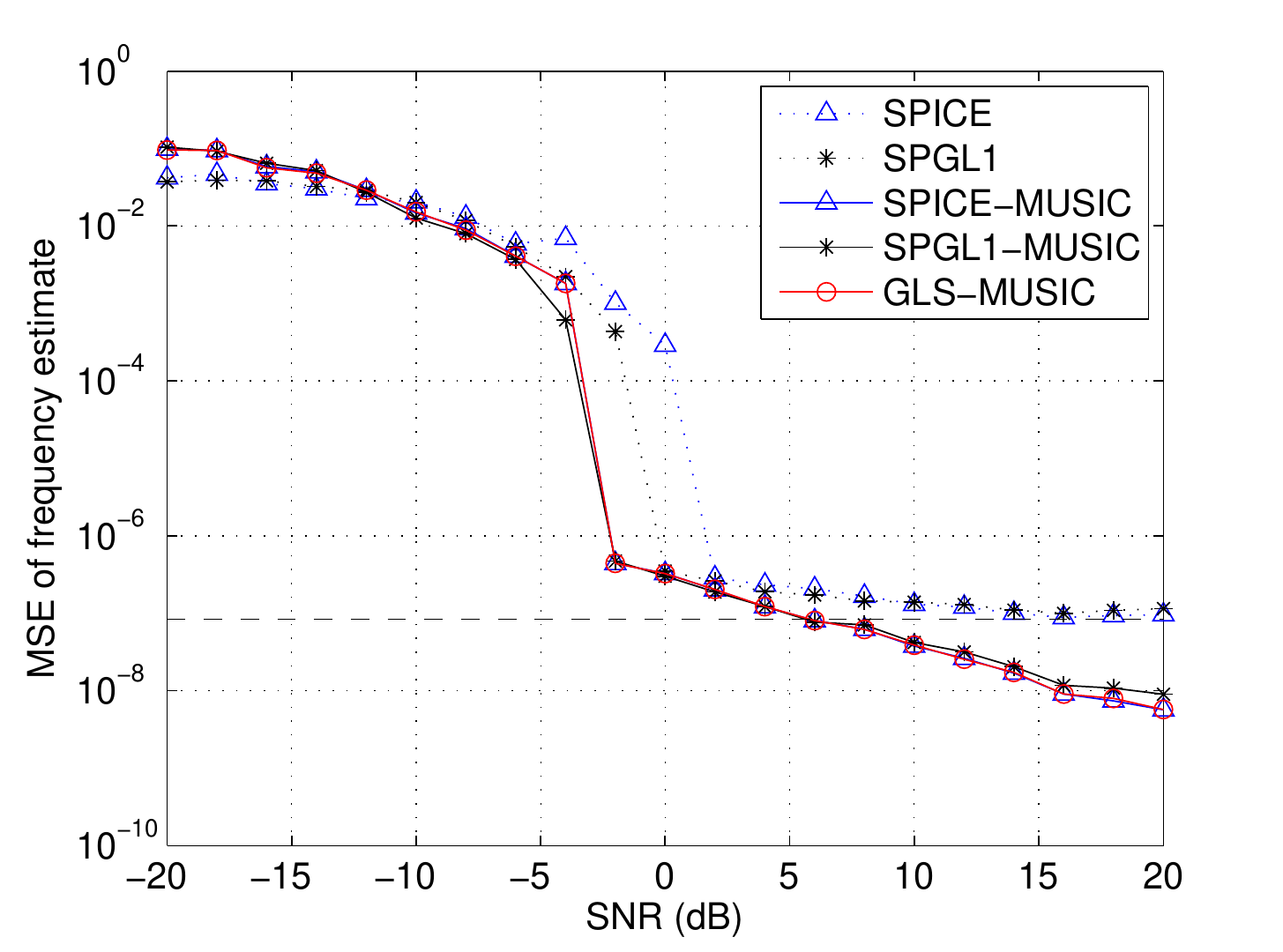}}
\centering
\caption{Results of model order selection (using SORTE) and frequency estimation with respect to SNR. The black dashed lines in (b) and (c) denote the lower bound $\frac{1}{12N^2}$.} \label{Fig:varySNR}
\end{figure*}

We study scalability of GLS in the following simulation. Besides GLS, SPICE and SPGL1, we consider another popular grid-based method named IAA. Differently from the sparse methods, IAA is based on a weighted least squares criterion. Since IAA does not explicitly optimize an objective function, a rigorous convergence analysis of IAA has not been available. In our setup, we proportionally increase the problem dimension. In particular, we let $M=50\kappa$, $L=30\kappa$ and $K=2\kappa$, and consider $\kappa=1,\dots,10$. Moreover, we randomly generate the frequencies with the minimum separation $\Delta_f\geq \frac{1}{M}$ and each power parameter as $1+w^2$, where $w$ is standard normal distributed. We fix the variance of i.i.d. Gaussian noise $\sigma=1$. We set the grid size $N=10M$ for SPICE, SPGL1 and IAA. 40 problems are generated and solved for each value of $\kappa$. The averaged computational times of the four methods are presented in Fig. \ref{Fig:varyDim_time}. Indeed, GLS is most time-consuming when the problem dimension increases due to the eigen-decomposition at each iteration. As a first-order method, the computational time of SPGL1 increases most slowly with the dimension. The performances in model order selection and frequency estimation are presented in Figs. \ref{Fig:varyDim_modelorder} and \ref{Fig:varyDim_MSEf_N5M}. It is shown that GLS and its grid-based versions SPICE and SPGL1 can accurately estimate the model order with only a few failures, and also perform well in frequency estimation within the proposed framework. In contrast, more failures of model order selection happen for IAA. Note that the seemingly bad performance of IAA-MUSIC in frequency estimation when $M\geq250$ is caused by very few outliers (similarly for SPICE and SPGL1 at some values of $M$).


\begin{figure*}
\centering
  \subfigure[Computational time]{
    \label{Fig:varyDim_time}
    \includegraphics[width=2.3in]{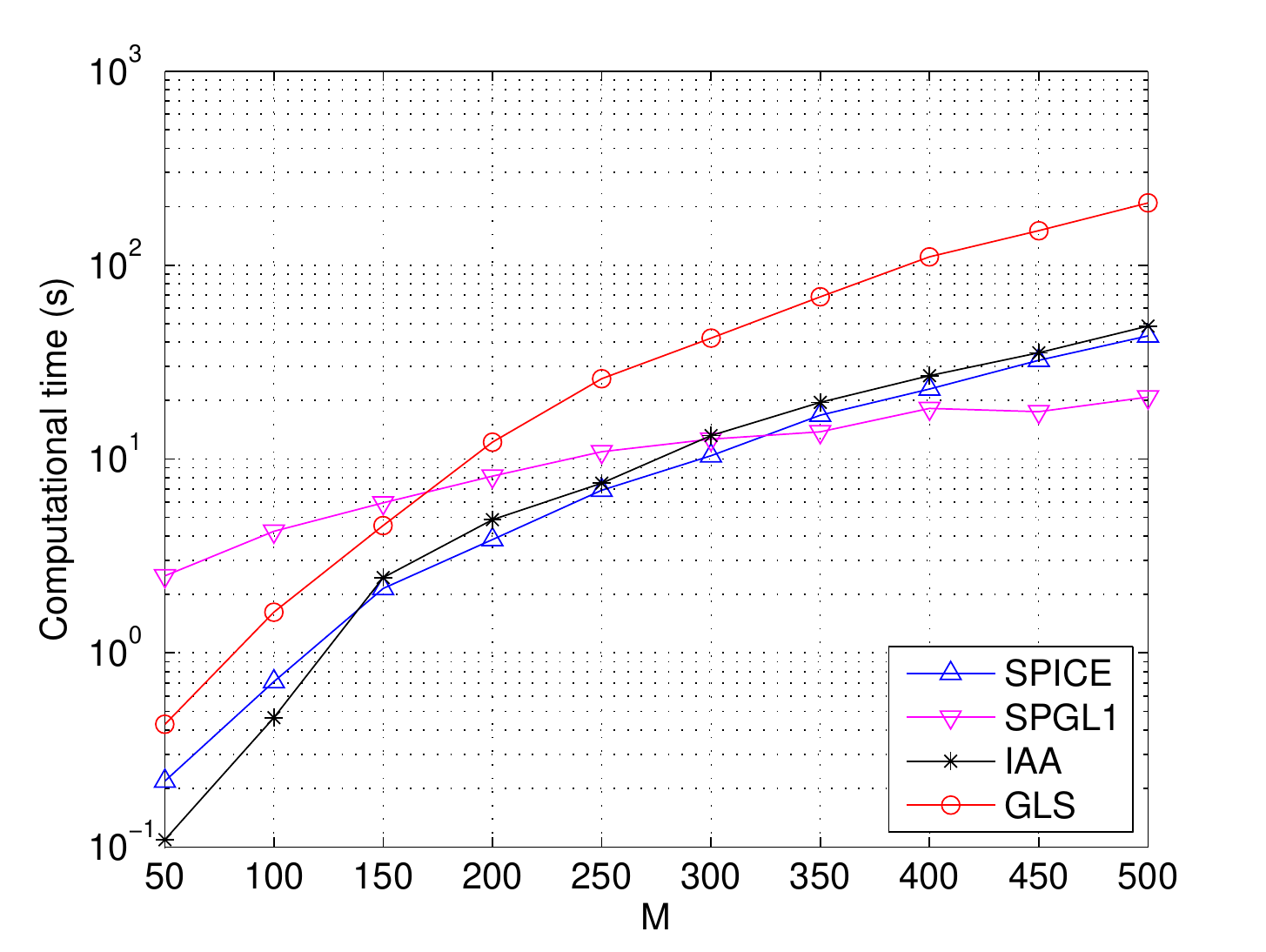}} %
  \subfigure[Model order selection]{
    \label{Fig:varyDim_modelorder}
    \includegraphics[width=2.3in]{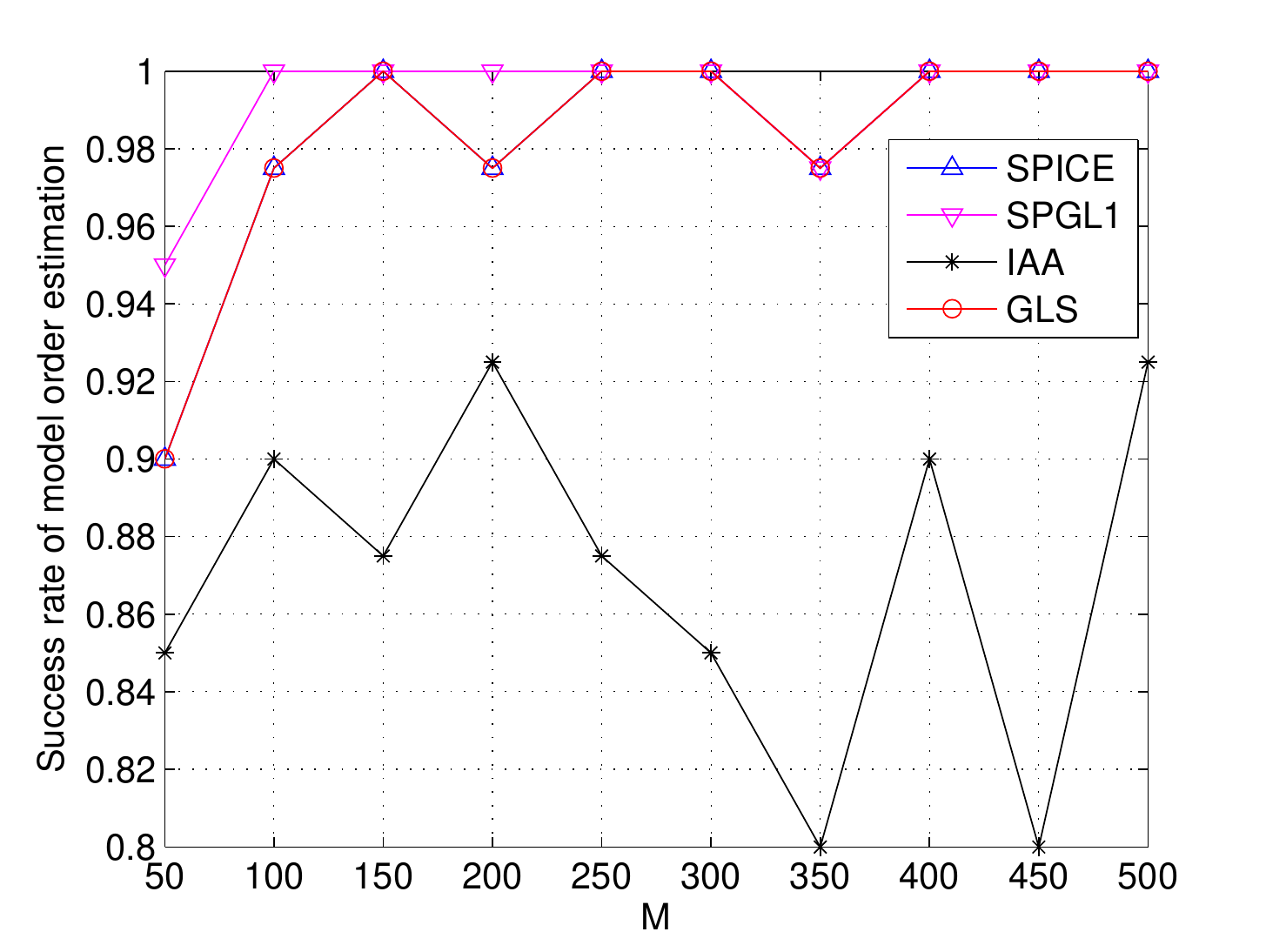}} %
  \subfigure[Frequency estimation]{
    \label{Fig:varyDim_MSEf_N5M}
    \includegraphics[width=2.3in]{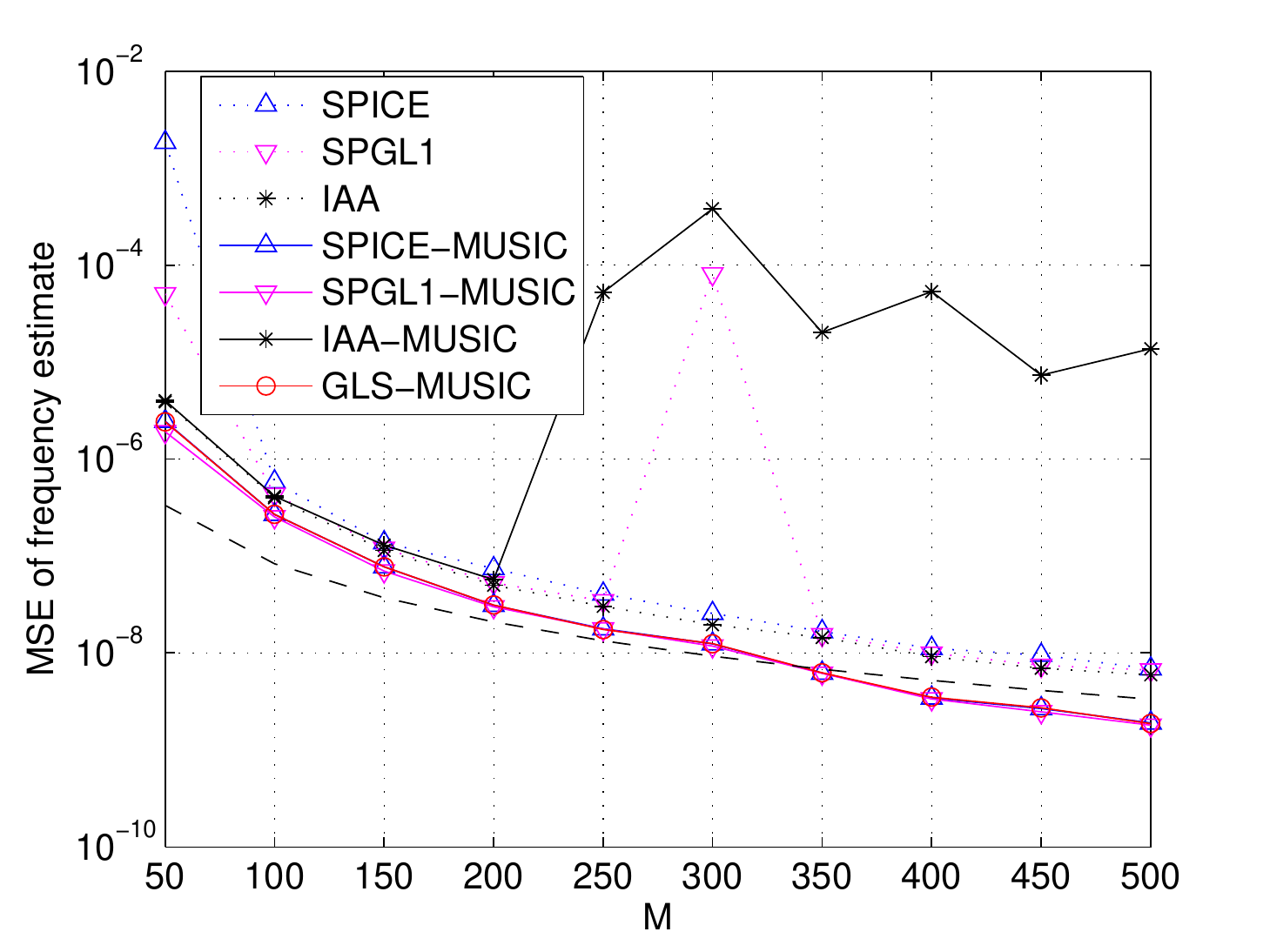}}
\centering
\caption{Computational time, model order selection (using SORTE) and frequency estimation with respect to problem dimension. The black dashed line in (c) denotes the lower bound $\frac{1}{12N^2}$, where $N=10M$ is used for SPICE, SPGL1 and IAA.} \label{Fig:varyDim}
\end{figure*}

We further compare the proposed model order selection method with conventional information-theoretic methods and show that the practical use of the latter is limited by the difficulty in solving the NLS problems. To do this, we measure the success rate of solving the NLS at the true value of $K$. We adopt the same experimental setup as in the previous simulation. An efficient algorithm for NLS is the expectation-majorization (EM) algorithm in \cite{feder1988parameter} in which the complicated multi-parameter NLS problem is decoupled into $K$ separate one-dimensional optimization problems at each iteration which can be efficiently solved by simple line search. The highest $K$ peaks of the periodogram are used to initialize the algorithm according to \cite{stoica2005spectral}. The NLS is considered to be successfully solved if $\inftyn{\widehat{\m{f}}-\m{f}}<\frac{1}{2M}$. Note that the criterion above is not stringent according to Fig. \ref{Fig:varyDim_MSEf_N5M}. Indeed, it is observed that all failures severely violate the criterion. The success rate measured over 100 Monte Carlo runs at each $M$ is presented in Fig. \ref{Fig:NLS}. It is shown that the EM algorithm is more likely to be trapped at a local optimum as the problem dimension increases. At $M=500$ and $K=20$ about one third of the NLS problems cannot be accurately solved. Therefore, we cannot expect that the MDL/AIC/BIC method gives faithful model order selection based on the NLS solution. In contrast, at most a single failure of model order selection out of 40 runs is observed using the proposed method when $M\geq100$ according to Fig. \ref{Fig:varyDim_modelorder}. The averaged computational time of the EM algorithm is also reported in Fig. \ref{Fig:NLS}. Note that a series of NLS problems need to be solved for model order selection, while the computational speed can be accelerated by implementing the $K$ optimization problems at each iteration in parallel.


\begin{figure}
\centering
  \includegraphics[width=2.8in]{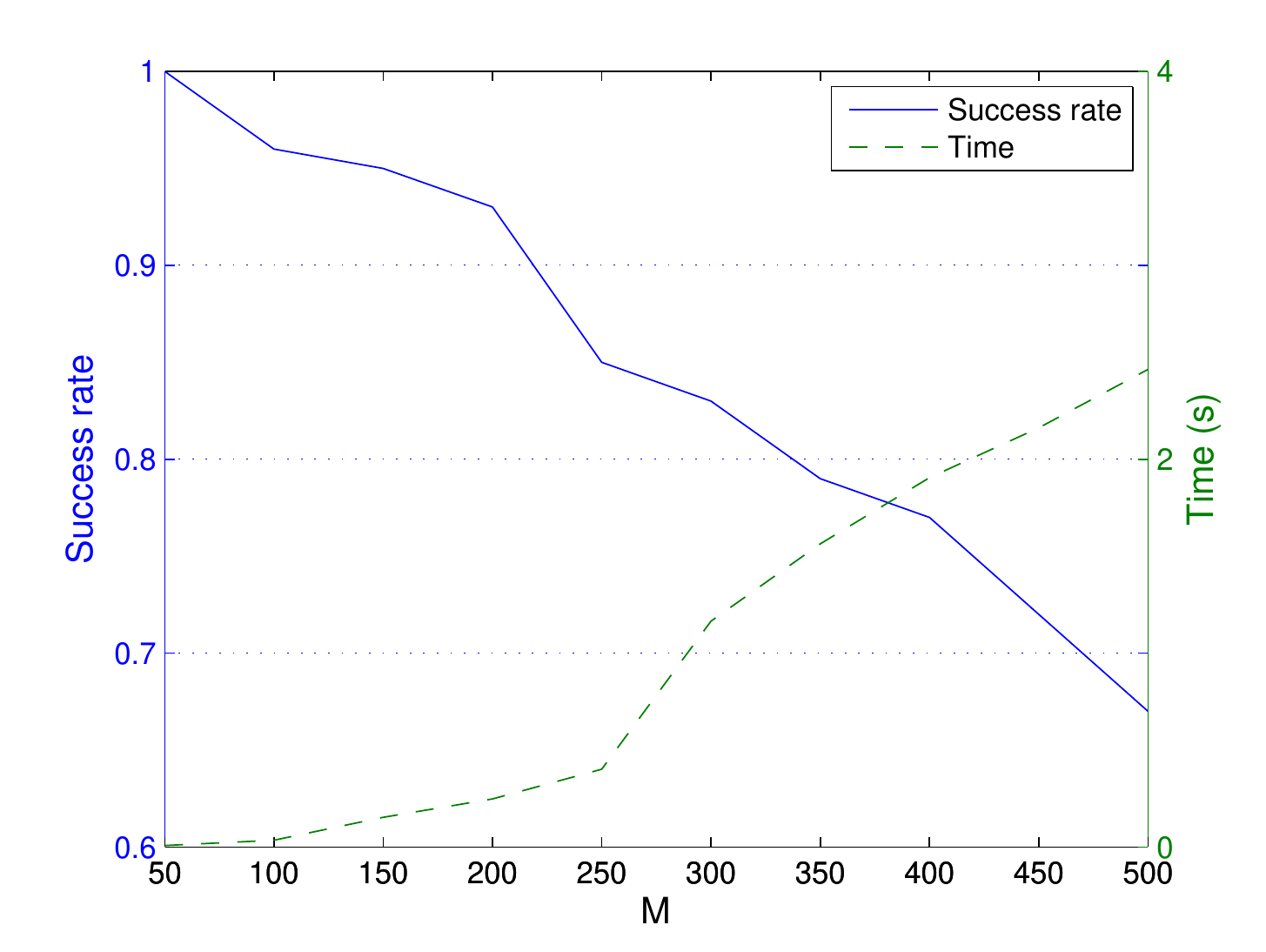}
\centering
\caption{Success rate and computational time of solving NLS using the EM algorithm in \cite{feder1988parameter} with the oracle information of $K$.} \label{Fig:NLS}
\end{figure}

To sum up, the proposed GLS and framework have good performances in both model order selection and frequency estimation in the presence of modest or light noise. Its grid-based versions SPICE and SPGL1 are generally good approximations with accelerated computations but might overestimate the model order with a less dense grid and suffer from convergence issues with a highly dense grid. Compared to conventional information-theoretic model order selection methods, the proposed method is of more practical interest.

\subsection{Resolution}

While the frequencies are separated by at least $\frac{1}{M}$ in previous simulations, we next study the capability of GLS in super-resolving two closely spaced frequencies compared to SPICE with $N=10M$. We fix $K=2$, $M=100$, $L=50$ and $\text{SNR}=10$dB, and vary the separation between the two frequencies, denoted by $\Delta_f$, from $\frac{0.1}{M}$ to $\frac{1}{M}$. Based on GLS (or SPICE), we estimate the two frequencies in three ways. In the first method, we simply select the largest two components of GLS (or the largest two peaks of SPICE). In the second, MUSIC is carried out after GLS (or SPICE) with the $K$ information. That is, the oracle information of $K$ is utilized in the first two methods. The last practical method strictly follows the proposed framework in which the model order used in MUSIC is given by SORTE. The two frequencies are considered to be successfully resolved if the absolute estimation error of every frequency is smaller than $\frac{1}{2}\Delta_f$. We measure the success rates of GLS and SPICE at each $\Delta_f$ over 100 Monte Carlo runs and present the results in Fig. \ref{Fig:resolution}. It is shown that both GLS and SPICE can super-resolve two closely spaced frequencies. GLS outperforms SPICE with the first method, especially at small values of $\Delta_f$ where gridding of the frequency interval exhibits more obvious drawbacks. Their performances are almost the same with the latter two methods since they produce slightly different covariance estimates. The proposed framework has good performance and its gap to the first two methods (with the oracle $K$) diminishes as the separation increases.

\begin{figure}
\centering
  \includegraphics[width=3in]{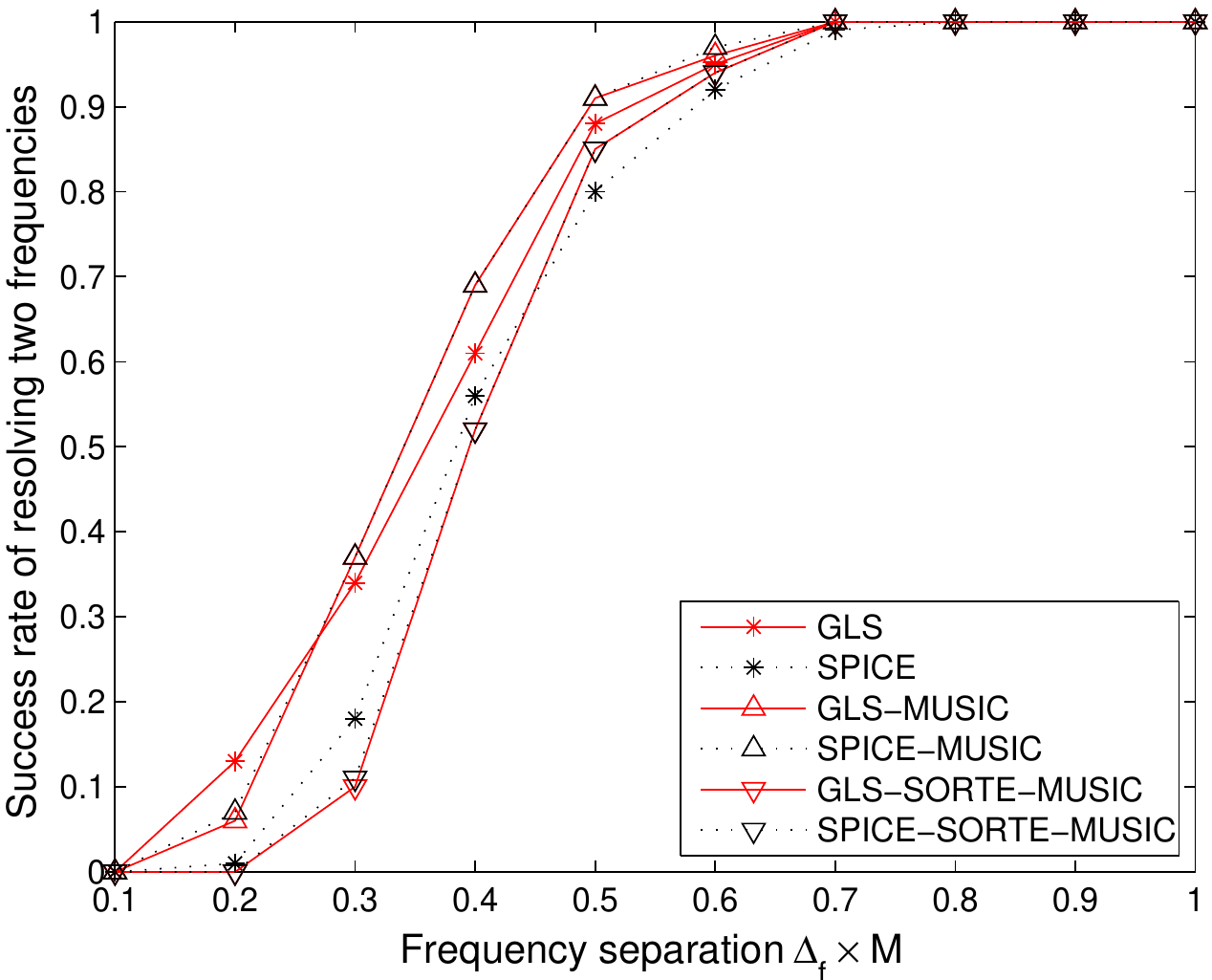}
\centering
\caption{Success rates of GLS and SPICE in resolving two frequencies separated by $\Delta_f$, with $K=2$, $M=100$, $L=50$ and $\text{SNR} = 10$dB.} \label{Fig:resolution}
\end{figure}

\section{Conclusion} \label{sec:conclusion}
The sparse, continuous frequency estimation problem was studied in this paper under the topic of line spectral estimation. Two gridless sparse methods were studied including the atomic norm based AST and weighted covariance fitting based GLS. Theoretical analysis of AST generalizes the existing result in the complete data case. GLS requires neither the model order nor the noise variance but might suffer from some limitations. A systematic framework consisting of model order selection and robust frequency estimation was proposed to overcome the limitations. Both AST and GLS were formulated as convex atomic norm denoising problems with practical algorithms proposed. Their performances were demonstrated on simulated data and compared to existing methods.

The first-order ADMM-based algorithms proposed for the gridless sparse methods are slow compared to existing grid-based ones since they converge slowly and need to carry out an eigen-decomposition at each iteration. A future work is to develop faster solvers for the SDPs involved in this paper. Inspired by a recent paper \cite{hsieh2011sparse} which shows that second-order solvers can be faster due to their fast convergence speed, we may turn to second-order algorithms in future studies. On the other hand, the framework extended from GLS is also applicable to its grid-based versions including SPICE and $\ell_1$ optimization, with satisfactory performances demonstrated in this paper. So, before emergence of very efficient solvers of GLS, its grid-based versions can be adopted as faster alternatives within the framework. Furthermore, it would be interesting to extend the proposed framework to other sparse parameter estimation problems in such as source localization and radar imaging.

\section*{Acknowledgment}
The authors are grateful to the reviewers for helpful comments which improve the content of this paper. Thanks to Petre Stoica and Dave Zachariah for sharing their code of IAA.

\appendix

\subsection{Vandermonde Decomposition and Its Realization} \label{sec:retrieval}
The Vandermonde decomposition is stated in the following lemma and its proof can be found in \cite{grenander1958toeplitz,stoica2005spectral}.
\begin{lem} Any positive semidefinite Toeplitz matrix $T\sbra{\m{u}}\in\bC^{M\times M}$ can be represented as
$T\sbra{\m{u}}=\m{A}\sbra{\m{f}}\m{P}\m{A}^H\sbra{\m{f}}$,
where $\m{A}\sbra{\m{f}}=\mbra{\m{a}\sbra{f_1},\cdots,\m{a}\sbra{f_K}}$, $\m{P}=\diag\sbra{p_1,\cdots,p_K}$, $f_j\in\left[0,1\right)$, $p_j>0$ for $j\in[K]$, and $K=\rank\sbra{T\sbra{\m{u}}}$. Moreover, the representation is unique if $K\leq M-1$. $\qed$ \label{lem:toeplitz}
\end{lem}

Lemma \ref{lem:toeplitz} states the existence and uniqueness of the Vandermonde decomposition of $T\sbra{\m{u}}\in\bC^{M\times M}$ provided that it is positive semidefinite and rank-deficient. We introduce a systematic method as follows to solve the parameters $\m{f}$ and $\m{p}\succ\m{0}$ satisfying that $T\sbra{\m{u}}=\m{A}\sbra{\m{f}}\diag\sbra{\m{p}} \m{A}^H\sbra{\m{f}}$. First, it is easy to show that
\equ{\begin{bmatrix}\m{A}\sbra{\m{f}}\\ \overline{\m{A}}_{\lbra{2,\cdots,M}}\sbra{\m{f}}\end{bmatrix}\m{p}= \begin{bmatrix}\overline{\m{u}} \\ \m{u}_{\lbra{2,\cdots,M}}\end{bmatrix}, \label{formu:linsys}}
where $\overline{\cdot}$ denotes the complex conjugate and $\overline{\m{A}}_{\lbra{2,\cdots,M}}\sbra{\m{f}}$ takes all but the first rows of $\overline{\m{A}}\sbra{\m{f}}$. Let $b_{j-1}=u_{j}$ and $b_{1-j}=\overline{u}_{j}$, $j\in\mbra{M}$ (note that $u_1\in\bR$). Then (\ref{formu:linsys}) can be written exactly as
\equ{b_m = \sum_{k=1}^K p_k\theta_k^m,\quad \theta_k=e^{-i2\pi f_k},\label{formu:b}}
for $1-M\leq m\leq M-1$. This system of equations can be solved using Prony's method (see, e.g., \cite{blu2008sparse}). We provide detailed procedures for completeness of this paper. Define the so-called annihilating filter with $z$-transform
\equ{H\sbra{z} = \sum_{k=0}^K h_kz^{-k} = \prod_{k=1}^K\sbra{1-\theta_k z^{-1}}, \label{formu:H}}
where $h_k$, $k=0,1,\dots,K$, are the filter coefficients with $h_0=1$. $H\sbra{z}$ is called the annihilating filter since it can be verified by (\ref{formu:b}) and (\ref{formu:H}) that
\equ{\begin{split}h_m \ast b_m
&=\sum_{k=0}^K h_kb_{m-k} =\sum_{k=0}^K h_k \sum_{l=1}^K p_l\theta_l^{m-k}\\
&=\sum_{l=1}^K p_l\theta_l^{m} H\sbra{\theta_l} =0. \end{split}\label{formu:hb0}}
Based on any $2K<2M-1$ consecutive values of $b_m$, we can build a linear system of $K$ equations by (\ref{formu:hb0}), from which the coefficients $h_k$, $k\in\mbra{K}$, can be solved. Then $\theta_k=e^{-i2\pi f_k}$, $k\in\mbra{K}$, are obtained as roots of $H(z)$ by (\ref{formu:H}). After that, $p_k$ can be easily solved based on (\ref{formu:b}).

\subsection{Proof of (\ref{formu:upperbound})} \label{sec:proofupperbound}
Our proof is inspired by \cite{bhaskar2013atomic} on the complete data case. By (\ref{formu:datomni}) the dual atomic norm of $\m{w}_{\m{\Omega}}\in\bC^L$ in the missing data case is given by
\equ{\datomni{\m{w}_{\m{\Omega}}}=\sqrt{L\sigma}\sup_{f\in\left[0,1\right)} \abs{W\sbra{e^{i2\pi f}}}, \label{formu:datomicnorm2}}
where
\equ{W\sbra{e^{i2\pi f}} = \frac{1}{\sqrt{L\sigma}} \sum_{m\in\m{\Omega}} w_me^{-i2\pi \sbra{m-1}f}}
is standard Gaussian distributed for any $f$ given that $w_m$ is Gaussian distributed with variance $\sigma$. Without loss of generality, we assume $\Omega_1=1$ and $\Omega_L=\overline{M}$ since the distribution of $\m{w}$ is invariant due to a constant phase change. For convenience, denote $\overline{W}=\sup_{f\in\left[0,1\right)} \abs{W\sbra{e^{i2\pi f}}}$ and similarly define $\overline{W'}$, where $W'$ denotes the derivative of $W$. Similarly to \cite{bhaskar2013atomic}, our proof is based on the following two results.
\begin{lem}[\cite{schaeffer1941inequalities}] Let $q(z)$ be any polynomial of degree $n$ on complex numbers with derivative $q'(z)$. Then,
\equ{\sup_{\abs{z}\leq1} \abs{q'\sbra{z}} \leq n\sup_{\abs{z}\leq1}\abs{q\sbra{z}}.} \label{lem:bernstein}
\end{lem}
\begin{lem}[\cite{bhaskar2013atomic}] Let $x_1,\cdots,x_N$ be complex Gaussian random variables with unit variance. Then,
\equ{E\mbra{\max_{1\leq n\leq N}\abs{x_i}} \leq \sqrt{\ln N+1}.} \label{lem:maxgauss}
\end{lem}

The derivation is divided into two steps. At the first step we show that the dual atomic norm expressed in (\ref{formu:datomicnorm2}) can be upper bounded by the maximum of its values on a uniform grid of $N$ points on the unit circle $\left[0,1\right)$. At the second step an upper bound of the maximum is computed and the grid number $N$ is optimized. According to Lemma \ref{lem:bernstein}, for any $f,s\in\left[0,1\right)$ it holds that
\equ{\begin{split}
&\abs{W\sbra{e^{i2\pi f}}}-\abs{W\sbra{e^{i2\pi s}}}\\
& \leq \abs{e^{i2\pi f}-e^{i2\pi s}} \overline{W'}\\
& = \abs{e^{i\pi\sbra{f+s}}\sbra{e^{i\pi\sbra{f-s}}-e^{i\pi\sbra{-f+s}}}} \overline{W'}\\
&\leq 2\pi\abs{f-s}\cdot \overline{M}\cdot\overline{W}. \end{split} \label{formu:W1}}
Let $s$ take values in the set $\lbra{0,\frac{1}{N},\cdots,\frac{ N-1}{N}}$. For any $f$, we may find some $s$ in the set such that $\abs{f-s}\leq \frac{1}{2N}$. It then follows from (\ref{formu:W1}) that
$\overline{W}\leq \max_{0\leq m\leq N-1}\abs{W\sbra{e^{i2\pi m/N}}} + \frac{\pi\overline{M}}{N}\overline{W}$
and thus
\equ{\overline{W}\leq \sbra{1-\frac{\pi\overline{M}}{N}}^{-1}\max_{0\leq m\leq N-1}\abs{W\sbra{e^{i2\pi m/N}}}. \label{formu:boundW}}

At the second step, it follows from (\ref{formu:boundW}) and Lemma \ref{lem:maxgauss} that
\equ{E\datomni{\m{w}} \leq \sqrt{L\sigma}\sbra{1-\frac{\pi\overline{M}}{N}}^{-1} \sqrt{\ln N+1}. \label{formu:edatomn2}}
Let $N=p\pi\overline{M}$. Then it is easy to show that the right hand side of (\ref{formu:edatomn2}) is minimized when $p$ is the limit point of the sequence $\lbra{p_j}$ with $p_0>2$ and
\equ{p_{k+1}=2\ln p_{k} + 2\ln \sbra{\pi\overline{M}}+3.}
Moreover, the limit point falls in the interval $\sbra{2\ln\overline{M},5\ln\overline{M}}$ when $\overline{M}$ is modestly large.


\subsection{Proof of (\ref{formu:Natomnvsatomn})} \label{sec:Natomnvsatomn}
For a column vector $\m{w}_{\m{\Omega}}\in\bC^L$, its dual atomic norms are
\lentwo{\equa{\datomni{\m{w}_{\m{\Omega}}}
&=& \sup_{f\in\left[0,1\right)} \abs{\sum_{m\in\m{\Omega}} w_me^{-i2\pi \sbra{m-1}f}}, \label{formu:datomicnorm21}\\ \Ndatomni{\m{w}_{\m{\Omega}}}
&=&\sup_{f\in\widetilde{\m{f}}} \abs{\sum_{m\in\m{\Omega}} w_me^{-i2\pi \sbra{m-1}f}}. \label{formu:datomicnorm22}}
}An immediate result is that
$\datomni{\m{w}_{\m{\Omega}}}\geq \Ndatomni{\m{w}_{\m{\Omega}}}$.
On the other hand, by (\ref{formu:boundW}) we have
$\datomni{\m{w}_{\m{\Omega}}}\leq \sbra{1-\frac{\pi\overline{M}}{N}}^{-1} \Ndatomni{\m{w}_{\m{\Omega}}}$.
So it holds that
\equ{\Ndatomni{\m{w}_{\m{\Omega}}}\leq \datomni{\m{w}_{\m{\Omega}}}\leq \sbra{1-\frac{\pi\overline{M}}{N}}^{-1} \Ndatomni{\m{w}_{\m{\Omega}}},}
which concludes (\ref{formu:Natomnvsatomn}).

\bibliographystyle{IEEEtran}


\end{document}